\documentclass[onecolumn,11pt,aps,nofootinbib,superscriptaddress,tightenlines,showkeys]{revtex4}

\usepackage[T1]{fontenc}
\usepackage[utf8]{inputenc}

\usepackage[cmex10]{amsmath}
\usepackage{graphics}
\usepackage{dsfont}
\usepackage{amssymb}
\usepackage{graphicx}
\usepackage{mathrsfs}
\usepackage{units}
\usepackage{pgfplots}
\usepackage{enumitem}
\usepackage[colorlinks=true,linkcolor=black,citecolor=black,plainpages=false,pdfpagelabels]{hyperref}
\hypersetup{pdftitle={Eigenvalue estimates for the resolvent of a non-normal matrix}}

\usepackage{amsthm}
\newtheorem{theorem}{Theorem}[section]
\newtheorem{lemma}[theorem]{Lemma}
\newtheorem{proposition}[theorem]{Proposition}
\newtheorem{corollary}[theorem]{Corollary}
\newtheorem{definition}[theorem]{Definition}

\newenvironment{remark}[1][Remark]{\begin{trivlist}
\item[\hskip \labelsep {\bfseries #1}]}{\end{trivlist}}

\newcommand{\h}{\ensuremath{\mathcal{H}}}

\newcommand{\R}{\mathbb{R}}

\newcommand{\id}{\ensuremath{\mathds{1}}}

\begin{document}
\title{The Role of Topology in Quantum Tomography}
\author{Michael Kech}
\email{kech@ma.tum.de}
\affiliation{Department of Mathematics, Technische Universit\"{a}t M\"{u}nchen, 85748 Garching, Germany}
\author{P\'{e}ter Vrana}
\email{vranap@math.bme.hu}
\affiliation{Department of Geometry, Budapest University of Technology and Economics, Egry J\'{o}zsef u. 1., 1111 Budapest, Hungary
}
\author{Michael M. Wolf}
\email{wolf@ma.tum.de}
\affiliation{Department of Mathematics, Technische Universit\"{a}t M\"{u}nchen, 85748 Garching, Germany}

\date{\today}

\begin{abstract}
We investigate quantum tomography in scenarios where prior information restricts the state space to a smooth manifold of lower dimensionality. By considering stability we provide a general framework that relates the topology of the manifold to the minimal number of binary measurement settings that is necessary to discriminate any two states on the manifold. 

We apply these findings to cases where the subset of states under consideration is given by states with bounded rank, fixed spectrum, given unitary symmetry or taken from a unitary orbit. For all these cases we provide both upper and lower bounds on the minimal number of binary measurement settings necessary to discriminate any two states of these subsets.
\end{abstract}
\keywords{quantum tomography, prior information, topology, smooth embeddings\vspace*{-25pt}}

\maketitle

\tableofcontents
\section{Introduction}
The reconstruction of a quantum state from the outcome of an experiment, called quantum state tomography, is a task of fundamental importance in quantum information science. Already for small systems this task may be non-trivial, requiring many measurements and extensive postprocessing to reconstruct a state. With growing  system size this complexity becomes exceedingly relevant \cite{haffner2005scalable}. 

There are at least three kinds of resources that can be considered in this context: (i) the number of measurement settings or, mathematically equivalent, the number of measurement outcomes if a single generalized measurement is considered, (ii) the number of samples to be measured, i.e., the sampling complexity, which takes the statistics into account and (iii) the classical post-processing that is required to interpret the data. In this work we will focus on (i).

We are interested in cases where prior information is available that effectively restricts the state space to a submanifold of lower dimensionality. This information may concern the rank of the density operator, its spectrum, symmetry, energy, associated particle number or other properties and combinations thereof.

The question behind our analysis is: what is the minimal number of binary measurement settings that is required to uniquely identify the state under the assumption that it is taken from the given submanifold? Motivated by the results in \cite{heinosaari2013quantum} our aim is a better understanding of the relation between the minimal number of required measurement settings and the topology of the considered submanifold. Such a relation is most clear in low dimensional examples: suppose the submanifold forms a Klein bottle. Then, although it is two-dimensional, it requires at least four binary measurement settings to identify every point since, loosely speaking, in less than four dimensions the Klein bottle has no realization without self-intersections.

This topological reasoning was introduced in \cite{heinosaari2013quantum} and there applied successfully for instance to the case of pure state quantum tomography. The latter has been a topic of active research in quantum information theory \cite{heinosaari2013quantum, weigert1992pauli, amiet1999reconstructing, amiet1998reconstructing2, finkelstein2004pure, MV13, CarmeliTeikoJussi1, CarmeliTeikoJussi2, flammia2005minimal, QCS2, QCS3}, closely related to the problem of phase retrieval \cite{PhaseRetrieval, GrossPhaseLift}.

On a Hilbert space of dimension $d$, the set of pure states is of dimension $2d-2$, whereas the set of all states is of dimension $d^{2}-1$. Consequently, in order to uniquely identify an arbitrary state, one has to at least perform $d^{2}-1$ different binary measurements, whereas in the case of pure states one can hope that $\mathcal{O}(d)$ measurements suffice to uniquely identify a state.  In \cite{heinosaari2013quantum} it was shown that to leading order $4d$ binary measurements are necessary and sufficient to identify pure states and the compressed sensing approach of \cite{gross2010quantum, QCS2, QCS3} provides an algorithm based on  $\mathcal{O}(dr\log(d))$ binary measurements with which a $d\times d$ matrix of rank $r$ can be reliably identified.

The approach we take in this paper extends the results of \cite{heinosaari2013quantum} and gives a general framework for the validity of the topological reasoning in quantum tomography. Thereby, we show that the approach is applicable in the presence of statistical fluctuations, imprecise prior information or inaccuracies in the implementation of the measurement set-up. Moreover, we provide a detailed analysis of a variety of old and new examples of submanifolds.

\textit{Outline.} We consider measurements as smooth maps from a smooth submanifold of states into Euclidean space. The methods we deploy to find bounds on the number of measurement outcomes necessary to identify a state of a given submanifold uniquely rely on the technical assumption that this smooth map is a smooth embedding.

In section \ref{sec1} we give an operational meaning of the smooth embedding assumption and we determine the relation of quantum tomography to the embedding problem in differential topology.

First, in subsection \ref{sec11}, we justify the smooth embedding assumption by relating it to properties one would generally require of measurements. More precisely, we give two natural notions of stability and we show that these are in fact equivalent to the measurement being a smooth embedding. In the sense of these stability properties our approach is robust with respect to noise.

Secondly, in subsection \ref{sec12}, we generalize the measurement scheme by allowing for measurements on several copies of a state. We then show that any smooth embedding can be approximated by these generalized measurements. This proves that asking for the minimal number of measurement that is needed to identify all states of a given submanifold of states is equivalent to asking for the minimal dimension in which this manifold can be embedded.

Having justified our methods of finding bounds in section \ref{sec1}, we devote section \ref{sec2} to applying this method in concrete scenarios. We obtain upper and lower bounds on the number of measurement outcomes necessary to identify states of certain interesting submanifolds. The lower bounds result from topological obstructions, whereas the upper bounds rely on the explicit construction of measurement schemes. The methods used in this section are very different from the the ones used in the first two sections and from this point of view section \ref{sec2} can be read independently. 

First, we investigate states of fixed spectrum and we relate these to states of bounded rank in subsection \ref{sec21}. More precisely, we present lower bounds on the number of measurements necessary to identify states of fixed spectrum and these lower bounds turn out to be very close to the upper bounds for states of bounded rank obtained in \cite{heinosaari2013quantum}. In this way we obtain good upper and lower bounds for both the states of fixed spectrum and the states with bounded rank. 

In subsection \ref{unisym}, we obtain lower and upper bounds for states with a unitary symmetry and we use this to obtain both lower and upper bounds for states of fixed spectrum with a unitary symmetry in subsection \ref{sec23}. 

Finally, in subsection \ref{sec24}, we obtain upper and lower bounds for states in a bipartite system that lie in the Bob-unitary orbit of a certain pure state, i.e. we consider all states that can be reached from a given pure state by acting with a unitary matrix that just effects  Bob's subsystem.  Physically, this scenario may correspond to an interferometry experiment. Note that if the initial state is maximally entangled, this orbit is the set of maximally entangled states which may be interesting in its own right. Identifying a maximally entangled state is equivalent to determining the unitary matrix that acted on Bob's subsystem. So this method can also be used for process-tomography of unitary time evolutions, complementing the results in \cite{unitary1,unitary2}. 

Proofs of technical results can be found in the appendix.

\section{Preliminaries}

Let $\mathcal{H}$ be a finite dimensional Hilbert space. We denote by $\mathcal{B}(\mathcal{H})$ the complex vector space of linear operators on $\mathcal{H}$. $H(\mathcal{H})$ denotes the real vector space of hermitian operators on $\mathcal{H}$ and $H(\mathcal{H})_{0}$ denotes the real vector space of traceless hermitian matrices, i.e. $H(\mathcal{H})_{0}:=\{h\in H(\mathcal{H}):\text{tr}(h)=0\}$. Throughout we consider these spaces as inner product spaces equipping them with the Hilbert-Schmidt inner product. Furthermore, $\mathcal{S}(\mathcal{H})$ will denote the set of quantum states on $\mathcal{H}$, i.e. $\mathcal{S}(\mathcal{H}):=\{\rho\in H(\mathcal{H}):\rho\geq 0, \text{tr}(\rho)=1\}$.

A positive operator valued measure (POVM) corresponds to a set of positive semidefinite operators $P:=\{P_{1},...,P_{m}\}$ in $H(\mathcal{H})$ such that
\begin{align*}
\sum_{i=1}^{m}P_{i}=\id_\h.
\end{align*}
An element of $P$ is called an effect operator. We define the dimension of $P$ by $\dim P=m-1$. In quantum mechanics, POVMs are used to describe general measurements \cite{holevo2011probabilistic, busch1995operational}.

There is an operator system
\footnote{An operator system $\sigma\subseteq\mathcal{B}(\mathcal{H})$ is a linear subspace such that $\id_{\mathcal{H}}\in\sigma$ and $\sigma^{\dagger}=\sigma$.} 
$\sigma_{P}$ associated to each POVM $P$ given by the complex linear span of the operators of $P$. For an operator system $\sigma$ denote by $\sigma^{\R}$ the real vector space of hermitian operators in $\sigma$, i.e. $\sigma^{\R}:=\{h\in\sigma:h^{\dagger}=h\}$ \footnote{Note that $\sigma^{\R}$ determines $\sigma$ uniquely.}. In the following we assume the effect operators of a POVM $P$ to be linearly independent over $\mathbb{C}$. Note that by this convention $\dim P=\dim \sigma_{P}-1$. To each operator system $\sigma$ one can associate the orthogonal projection from $H(\h)$ to $\sigma^{\R}\subseteq H(\h)$. Throughout we denote this associated projection by $\pi_{\sigma}$.\\

\begin{definition}
A POVM $P:=\{P_{1},...,P_{m}\}$ induces a linear map
\begin{align*}
h_{P}:H(\mathcal{H})&\to \mathbb{R}^{m} \\
  \rho&\mapsto\big( \text{tr}(P_{1}\rho),...,\text{tr}(P_{m}\rho) \big).
\end{align*}
$P$ is called $\mathcal{R}$-complete for a subset $\mathcal{R}\subseteq \mathcal{S}(\mathcal{H})$ if $h_{P}|_{\mathcal{R}}$ is injective and it is called $\mathcal{P}$-embedding for a smooth submanifold $\mathcal{P}\subseteq \mathcal{S}(\mathcal{H})$ if $h_{P}|_{\mathcal{P}}$ is a smooth embedding \footnote{A smooth mapping $\psi:M\to N$ is called a smooth embedding if $d\psi_{x}$ is injective for all $x\in M$ and $\psi$ is a homeomorphism onto its image.}.
\end{definition}

Recall that the question behind our analysis concerns the minimal $m$ for which there is a $\mathcal{P}$-complete POVM for a given smooth submanifold  $\mathcal{P}\subseteq \mathcal{S}(\mathcal{H})$ that characterized the available prior information. From the dimension $D:=\dim\mathcal{P}$ alone one obtains that $m\geq D$ is necessary and $m=2D+1$ is generally sufficient \cite{heinosaari2013quantum}. For better bounds, one has to invoke more of the (topological) structure of the manifold.

In the following all manifolds and submanifolds are assumed to be smooth. Throughout we regard both $\mathcal{S}(\mathcal{H})$ and submanifolds $\mathcal{P}\subseteq\mathcal{S}(\mathcal{H})$ with $\mathcal{H}\simeq\mathbb{C}^{n}$ as submanifolds of $H(\h)\simeq\R^{n^2}$ equipped with the subspace topology and the standard smooth structure. We often use this picture to identify the tangent space at a point $\rho\in \mathcal{P}$, $T_{\rho}\mathcal{P}$, with a linear subspace in $H(\h)$, i.e. we think of tangent vectors $v\in T_{\rho}\mathcal{P}$ as hermitian operators. We assume submanifolds $\mathcal{P}\subseteq \mathcal{S}(\mathcal{H})$ to be closed and without boundary. In particular, this means that $\mathcal{P}$ is an embedded submanifold by the compactness of $\mathcal{S}(\mathcal{H})$, i.e. the inclusion is a homeomorphism onto its image. 

\section{Topological Analysis of Measurements}\label{sec1}
\subsection{Stable Measurements}\label{sec11}
Let $\mathcal{P}\subseteq \mathcal{S}(\mathcal{H})$ be a submanifold. In order for our methods for finding bounds on the dimension of $\mathcal{P}$-complete POVMs to apply, we need the technical requirement that these POVMs are $\mathcal{P}$-embeddings. In this section we justify this assumption. We develop two notions of stability for a $\mathcal{P}$-complete POVM and we show that these notions are equivalent to the POVM being a $\mathcal{P}$-embedding. These notions of stability are properties one would naturally require for $\mathcal{P}$-complete POVMs. Thus, under the premise of stability, $\mathcal{P}$-complete POVMs are $\mathcal{P}$-embeddings.\\
For a given POVM $P$ the notions of $\mathcal{P}$-embedding and $\mathcal{P}$-completeness just depend on its associated operator system $\sigma_{P}$ as the following proposition shows.

\begin{proposition}\label{propinj}
Let $P$ be a POVM, $h_{P}$ be the associated linear map, $\mathcal{P}\subseteq \mathcal{S}(\mathcal{H})$ be a submanifold and let $\pi_{P}:H(\mathcal{H})\to\sigma^{\R}_{P}$ be the orthogonal projection on $\sigma_{P}^{\R}$. \\
1. $P$ is $\mathcal{P}$-complete if and only if $\pi_{P}|_{\mathcal{P}}$ is injective.\\
2. $P$ is a $\mathcal{P}$-embedding if and only if $h_{P}$ is $\mathcal{P}$-complete and $d(\pi_{P}|_\mathcal{P})_{\rho}=\pi_{P}|_{T_{\rho}P}$ is injective for each $\rho\in\mathcal{P}$.
\end{proposition}
\begin{proof}
Since we equipped $H(\mathcal{H})$ with the Hilbert-Schmidt inner product, by the definition of $h_{P}$, we get $\sigma_{P}^{\R}=\text{span}_{\R}P\subseteq \ker (h_{P})^{\bot}$ and since $\text{rank} h_P=\dim \sigma_{P}^{\R}$, we get $\sigma_{P}^{\R}= \ker (h_{P})^{\bot}$ by dimensional reasons. So $h_{P}=h_{P}\circ\pi_{P}$.

For the first statement, let $P$ be $\mathcal{P}$-complete and $h_{P}$ be the associated linear map. Since $h_{P}|_{\mathcal{P}}$ is injective and $h_{P}=h_{P}\circ\pi_{P}$, we get that $\pi_{P}|_{\mathcal{P}}$ is injective.

Conversely, let $\pi_{P}|_{\mathcal{P}}$ be injective. Then, since $h_{P}=h_{P}\circ\pi_{P}$, $h_{P}|_{\mathcal{P}}$ is injective because $\pi_{P}|_{\mathcal{P}}$ is injective and $h_{P}$ is injective restricted to the image of $\pi_{P}$.\\
Noting that by linearity we have $d(h_{P}|_{\mathcal{P}})_{\rho}=dh_{P}|_{T_{\rho}\mathcal{P}}=h_{P}|_{T_{\rho}\mathcal{P}}$, the above reasoning also applies for the second statement.
\end{proof}
For a submanifold $\mathcal{P}\subseteq\mathcal{S}(\mathcal{H})$, $\mathcal{P}$-completeness and being a $\mathcal{P}$-embedding are the only properties of a POVM we are interested in. Thus, by proposition \ref{propinj}, there is a natural equivalence relation on the set of POVMs, namely
\begin{align*}
P\sim P^{\prime}\ \ \ \Leftrightarrow\ \ \ \sigma_{P}=\sigma_{P^{\prime}}.
\end{align*}
Since every $n$-dimensional operator system is generated by an $(n-1)$-dimensional POVM \cite{heinosaari2013quantum}, the operator systems are precisely the equivalence classes.

Since the proofs we give are easier to formulate using operator system we often state our results in terms of operator systems and then transfer them to POVMs.

Let $\Sigma (n)$ be the set of $n$-dimensional operator systems.  For a subset $\mathcal{R}\subseteq\mathcal{S}(\mathcal{H})$ we call $\sigma\in\Sigma(n)$ $\mathcal{R}$-complete if $\pi_{\sigma}|_{\mathcal{R}}$ is injective and for a submanifold $\mathcal{P}\subseteq\mathcal{S}(\mathcal{H})$ we call $\sigma\in\Sigma(n)$ a $\mathcal{P}$-embedding if $\pi_{\sigma}|_{\mathcal{P}}$ is a smooth embedding.
\\
 A metric on $\Sigma (n)$, which is natural for our purpose, can be defined in terms of any norm on the corresponding linear map $\pi_\sigma$. For an arbitrary linear map $L:H(\h)\to H(\h)$ we consider
\begin{align*}
\|L\|_{op}=\sup_{B\in H(\mathcal{H}),\|B\|\leq 1}\|L(B)\|
\end{align*}
where $\|\cdot\|$ denotes the Hilbert-Schmidt norm. The sought metric is then given by $d(\sigma,\sigma^{\prime})=\|\pi_{\sigma}-\pi_{\sigma^{\prime}}\|_{op}$. The following definition refers to the metric topology induced on $\Sigma (n)$.
\begin{definition}(Stability).
Let $\mathcal{R}\subseteq\mathcal{S}(\mathcal{H})$ be a subset. An $\mathcal{R}$-complete operator system $\sigma\in\Sigma (n)$ is called stably $\mathcal{R}$-complete if there exists a neighbourhood $N\subseteq\Sigma (n)$ of $\sigma$ such that every $\sigma^{\prime}\in N$ is an $\mathcal{R}$-complete operator system. A POVM $P$ is called stably $\mathcal{R}$-complete if its associated operator system $\sigma_{P}$ is $\mathcal{P}$-complete .
\end{definition}
\begin{remark}
In the following we will see that closeness of POVMs is equivalent to closeness of the associated operators systems. Thus, this definition says that a stably $\mathcal{P}$-complete POVM $P$ is robust against inaccuracy in its implementation in the sense that every close enough POVM is also $\mathcal{P}$-complete.
\end{remark}
The intuition behind this definition is best envisioned by thinking of operator systems as planes in $H(\h)\simeq \R^{d^2}$, see figure \ref{fig1}.
\begin{figure}[ht]
\includegraphics[width=7cm]{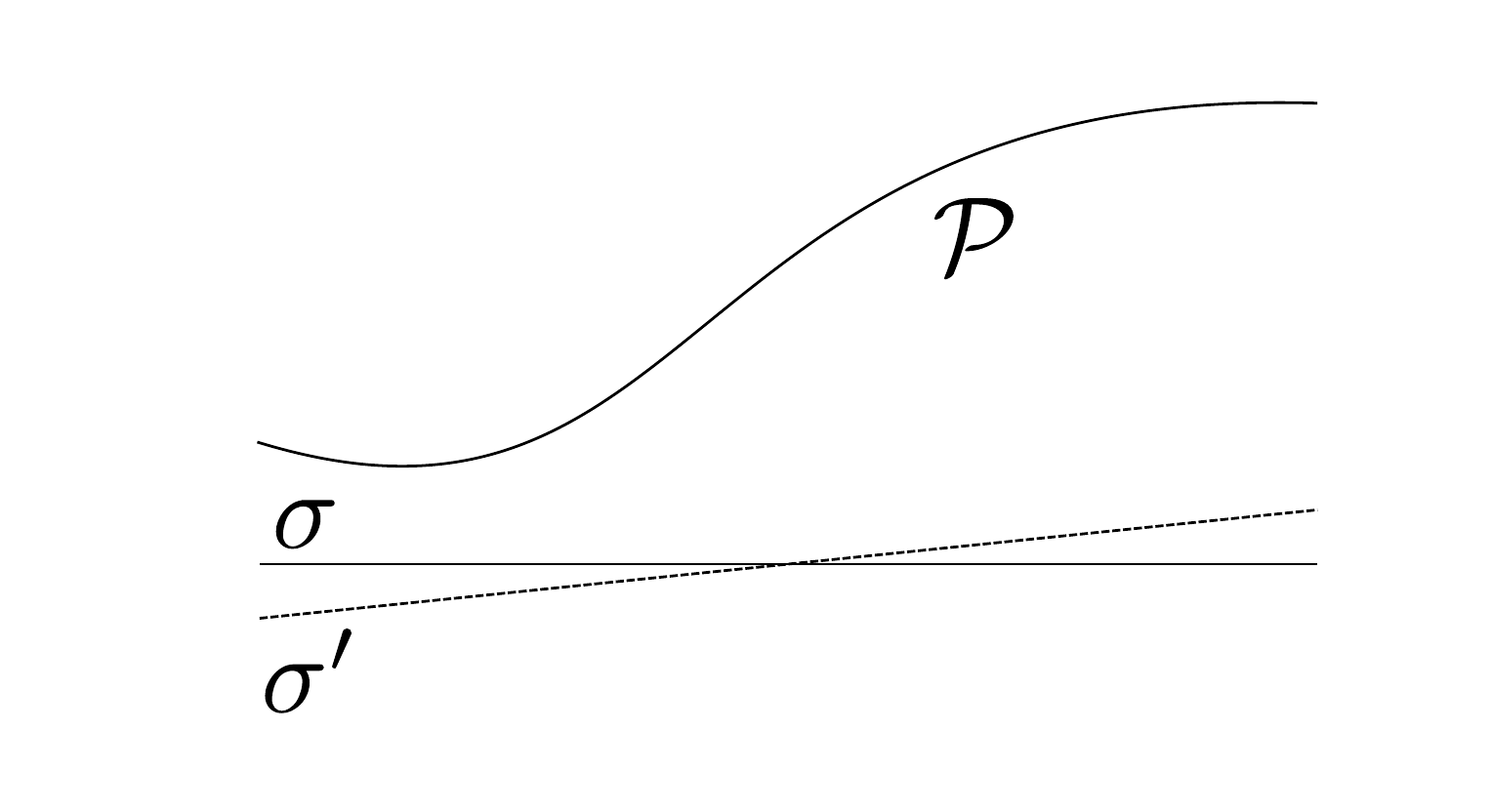}

\caption{\label{fig1}This figure shows a submanifold $\mathcal{P}\subseteq H(\h)\simeq \R^{d^2}$ and two close operator systems $\sigma,\sigma^{\prime}\subseteq H(\h)\simeq \R^{d^2}$ that are $\mathcal{P}$-complete.}
\end{figure}

Now we are in a position to state one of the main results of this section.
\begin{theorem}\label{thmstability}(Stable measurements are embeddings).
Let $\mathcal{P}\subseteq\mathcal{S}(\mathcal{H})$ be a closed submanifold and let $\sigma\in\Sigma (n)$ be $\mathcal{P}$-complete. Then $\sigma$ is stably $\mathcal{P}$-complete if and only if it is a $\mathcal{P}$-embedding.
\end{theorem}
Since the proof of this theorem is rather lengthy we relegated it to appendix \ref{appendixthmstability}.

Theorem \ref{thmstability} is a statement about operator systems. In order to provide it with an operational meaning, we prove the corresponding stability result for POVMs in the following.

\begin{corollary}
Let $\mathcal{P}\subseteq\mathcal{S}(\mathcal{H})$ be a closed submanifold and let $P:=\{P_1,...,P_m\}$ be a $\mathcal{P}$-complete POVM of dimension $m-1$ with associated linear map $h_P$. $P$ is a $\mathcal{P}$-embedding if and only if there is an $\epsilon>0$ such that every POVM $Q$ with $\sup_{v\in H(\h),\|v\|\leq 1}\|(h_{P}-h_{Q})v\|<\epsilon$ is $\mathcal{P}$-complete.
\end{corollary}

\begin{proof}
Let $\mathcal{H}=\mathbb{C}^{n}$ and let $\pi_{P}$ be the orthogonal projection associated to $P:=\{P_{1},...,P_{m}\}$.

Let $\epsilon>0$ such that every POVM $Q$ with $\sup_{v\in H(\h),\|v\|\leq 1}\|(h_{P}-h_{Q})v\|_{2}<\epsilon$ is $\mathcal{P}$-complete. We show, that there is $\delta>0$ such that for all $\sigma^{\prime}\in \Sigma(m)$ with
\begin{align*}
\|\pi_{P}-\pi_{\sigma^{\prime}}\|_{op}<\delta
\end{align*}
there is a POVM $P^{\prime}:=\{P_{1}^\prime,...,P_{m}^\prime\}$ with $\sigma_{P^{\prime}}=\sigma^{\prime}$ and $\sup_{v\in H(\h),\|v\|\leq 1}\|(h_{P}-h_{P^{\prime}})v\|_{2}<\epsilon$, hence $\pi_{\sigma^{\prime}}$ is $\mathcal{P}$-complete by proposition \ref{propinj}.

For every $\eta>0$, we can slightly deform $P$ to a POVM $\tilde{P}:=\{\tilde{P}_{1},...,\tilde{P}_{m}\}$ with full rank effect operators such that $\|P_{i}-\tilde{P}_{i}\|<\eta$ and $\sigma_{P}=\sigma_{\tilde{P}}$: Let $\tilde{P}_i:=\frac{\sqrt{n}}{\sqrt{n}+\eta}(P_i+\frac{\eta}{\sqrt{n}} \id_{\mathcal{H}})$ for $i=1,...,m$. Then, for $i=1,...,m$, $\|P_i-\tilde{P}_i\|=\eta\frac{\sqrt{n}}{\sqrt{n}+\eta}<\eta$. Note that we also ensured that the smallest eigenvalue of $\tilde{P}_{i}$ is bigger than $\eta/2$ for $i=1,...,m$ and $\eta$ small enough.

For some $\sigma^{\prime}\in \Sigma(m)$ with
\begin{align*}
\|\pi_{P}-\pi_{\sigma^{\prime}}\|_{op}<\delta
\end{align*}
let $P^{\prime}_{i}:=\pi_{\sigma^{\prime}}(\tilde{P}_{i})$, $i=1,...,m$. Then, for $i=1,...,m$,
\begin{align*}
\|\tilde{P}_{i}-P^{\prime}_{i}\|=\|\pi_{P}(\tilde{P}_{i})-\pi_{\sigma^{\prime}}(\tilde{P}_{i})\|<\sqrt{n}\delta
\end{align*}
and thus the $P^{\prime}_{i}$ are positive for $\sqrt{n}\delta<\eta/2$. Furthermore, $\sum_{i=1}^{m}P^{\prime}_{i}=\pi_{\sigma^{\prime}}(\sum_{i=1}^{m}\tilde{P}_{i})=\pi_{\sigma^{\prime}}(\id_{\mathcal{H}})=\id_{\mathcal{H}}$. For small enough $\delta$, the $P^{\prime}_{i}$ are linearly independent because the $P_{i}$ are linearly independent by assumption. Thus $\sigma^{\prime}=\sigma_{P^{\prime}}$ by dimensional reasoning. Finally, 
\begin{align*}
&\|h_{P}-h_{P^{\prime}}\|_{op}^{2}\\
\leq& \sum_{i=1}^{m}\|P_{i}-P_{i}^{\prime}\|^{2}=\sum_{i=1}^{n}\|P_{i}-\tilde{P}_{i}+\tilde{P}_{i}-P_{i}^{\prime}\|^{2}\\
\leq&\sum_{i=1}^{m}(\|P_{i}-\tilde{P}_{i}\|+\|\tilde{P}_{i}-P_{i}^{\prime}\|)^{2}\\
<&m(\eta^{2}+n\delta^2).
\end{align*}
By choosing $\eta$ and $\delta$ so small that $m(\eta^{2}+n\delta^2)<\epsilon^2$, $h_{P^{\prime}}$ is injective by assumption and thus $\sigma^{\prime}$ is injective by \ref{propinj}. Thus, $\sigma_{P}$ is stably $\mathcal{P}$-complete and \ref{thmstability} concludes the proof of this direction.

Conversely, suppose $\sigma$ is a $\mathcal{P}$-embedding. Corollary \ref{linap} states, that there is an $\epsilon>0$ such that every POVM $Q$ with $\sup_{B\in H(\mathcal{H}),\ \|B\|\leq 1}\|h_{P}(B)-h_{Q}(B)\|<\epsilon$
is a $\mathcal{P}$-embedding and thus in particular $\mathcal{P}$-complete.
\end{proof}

However, the notion of stability for measurements developed so far may not be satisfactory yet since it just considers inaccuracy in the implementation of the measurement set-up. Noisiness of the outcome, resulting from e.g. dissipation or finite statistics, or noisiness of the input, originating from e.g. inaccurate prior information, are inevitable but not considered in the definition.

In the remainder of this section we show that also from this point of view, stably $\mathcal{P}$-complete is a operationally meaningful property.

The idea of the following lemma, which is the essential ingredient for the second theorem of this section, is to construct a neighbourhood for every point of submanifold $\mathcal{P}\subseteq\mathcal{S}(\h)$ that can be approximated by the tangent space at that point. Let $\pi^{T}:\mathcal{P}\to \mathcal{B}(H(\mathcal{H}))$ be the mapping that associates to each point $\rho\in\mathcal{P}$ its orthogonal projection $\pi^{T}_{\rho}$ to $T_{\rho}\mathcal{P}\subseteq H(\h)$ and let $\pi^{N}$ be the analogue mapping for the normal space, i.e. $\pi^{N}_{\rho}+\pi^{T}_{\rho}=\text{id}_{H(\h)}$ for all $\rho\in\mathcal{P}$. Furthermore, $B_{x}(\epsilon)$ denotes the open ball with center $x\in H(\mathcal{H})$ and radius $\epsilon>0$, i.e. $B_{x}(\epsilon):=\{y\in H(\mathcal{H}):\|x-y\|<\epsilon\}$ and $d$ denotes the metric induced by $\|\cdot\|$.
\begin{lemma}\label{lemtangentapp}
Let $\mathcal{P}\subseteq\mathcal{S}(\mathcal{H})$ be a submanifold. For every $\eta>0$ there is an $\epsilon>0$ such that for all $\rho\in\mathcal{P}$,
\begin{align*}
\rho^{\prime}\in B_{\epsilon}(\rho)\cap\mathcal{P}\Rightarrow\|\pi^{N}_{\rho}(\rho^{\prime}-\rho)\|<\eta\|\rho-\rho^{\prime}\|.
\end{align*}
\end{lemma}
 Since the proof of this lemma in rather technical, it is relegated to appendix \ref{appendixlemtangentapp}.

The following theorem is the second main result of this section. It is formulated in terms of operator systems but again the result transfers to POVMs. Since the interpretation of the theorem may not be obvious let us first give some intuition and motivation: Adding small perturbations to states of a submanifold $\mathcal{P}\subseteq\mathcal{S}(\h)$ can be thought of as blowing up $\mathcal{P}$ to a small tubular neighbourhood $\mathcal{P}_\epsilon=\{\rho\in\mathcal{S}(\h):d(\rho,\mathcal{P})<\epsilon\}$. The dimension of $\mathcal{P}_\epsilon$ is then equal to the dimension of $H(\h)$. Thus, one cannot expect $\mathcal{P}$-complete POVMs to stay injective when allowing for small errors. However, one can hope for being able to separate points in $\mathcal{P}_\epsilon$ that are sufficiently far away, in the sense that $\pi_P(\rho)\neq\pi_P(\rho^{\prime})$ for $\|\rho-\rho^{\prime}\|>C\epsilon$ with $C>0$ a constant. For a given small enough $\epsilon$ such a $C$ obviously exists however it is not immediate that $C$ can be chosen independent of $\epsilon$. The following theorem asserts that for $\epsilon$ smaller than a certain fixed value, $C$ is independent of $\epsilon$. The existence of a  $C$ independent of $\epsilon$ means that the measurement can be made arbitrarily precise by reducing the errors.

\begin{theorem}\label{propstab}
Let $\mathcal{P}\subseteq \mathcal{S}(\mathcal{H})$ be a submanifold and let $P$ be a POVM with associated orthogonal projection $\pi_{P}$. $P$ is stably $\mathcal{P}$-complete if and only if there exists $\epsilon_{0}>0$ and $C>2$ such that for all $\epsilon$ with $0<\epsilon<\epsilon_{0}$,
\begin{align*}
\rho,\rho^{\prime}\in\mathcal{P}\text{  with }\|\rho-\rho^{\prime}\|>\epsilon C\ \Rightarrow\ \pi_{P}(B_{\epsilon}(\rho))\cap\pi_{P}(B_{\epsilon}(\rho^{\prime}))=\emptyset\ \ \ .
\end{align*}
\end{theorem}
\begin{proof}
Let $P$ be stably $\mathcal{P}$-complete and thus a $\mathcal{P}$-embedding by theorem \ref{thmstability} and let $\rho\in\mathcal{P}$. Let $l:=\max_{\rho\in\mathcal{P}}\|\pi_{P}\circ\pi_{\rho}^{T}-\pi_{\rho}^{T}\|_{op}$ and let $\eta$ as well as $\tilde{\epsilon}$ be as in lemma \ref{lemtangentapp}. Shrink $\epsilon_{0}$ such that $C\epsilon_{0}<\tilde{\epsilon}$.

Note that $l<1$ because $\pi_{P}$ is an immersion. Without the immersion property we could not assume $l<1$ and in fact this is the essential idea of this proof \footnote{The bigger $l$ is, the steeper the tangent spaces can be with respect to the operator system $\sigma_P$. By the previous lemma we saw that small neighbourhoods around a point $\rho\in\mathcal{P}$ can be  approximated by the tangent space at that point. We can ensure that these approximations are so good that the fluctuations around the steepest tangent space have no component orthogonal to $\sigma_P$ and in this sense we can locally think of $\mathcal{P}$ as a plane.}. 

Noting that $B=\{(\rho,\rho^{\prime})\in\mathcal{P}^{2}:\|\rho-\rho^{\prime}\|\geq\epsilon_{0}C\}$ is compact, $\kappa:=\min_{(\rho,\rho^{\prime})\in B}\|\pi_{P}(\rho-\rho^{\prime}\|$ is attained and thus $\kappa>0$ by the injectivity of $\pi_{P}$. If necessary, shrink $\epsilon_{0}$ such that  $\epsilon_{0}<\kappa/2$. Then, for $\rho^{\prime}\in\mathcal{P}-(B_{C\epsilon_{0}}(\rho)\cap\mathcal{P})$ the claim holds because $\|\pi_{P}(\rho-\rho^{\prime})\|\geq\kappa> 2\epsilon_{0}$ and $\pi_{P}(B_{\epsilon}(\rho))=B_{\epsilon}(\pi_{P}(\rho))$.

Finally, let $\rho^{\prime}\in B_{C\epsilon_{0}}(\rho)\cap\mathcal{P}$ and $\|\rho-\rho^{\prime}\|>C\epsilon$, $0<\epsilon<\epsilon_{0}$. Then,
\begin{align*}
\|\pi_{P}(\rho-\rho^{\prime})\|&\geq\|\pi_{P}(\pi^{T}_{\rho}(\rho-\rho^{\prime}))\|-\|\pi_{P}(\pi^{N}_{\rho}(\rho-\rho^{\prime}))\|\\
&\geq\|\pi_{P}(\pi^{T}_{\rho}(\rho-\rho^{\prime}))-\pi^{T}_{\rho}(\rho-\rho^{\prime})+\pi^{T}_{\rho}(\rho-\rho^{\prime})\|-\eta\|\rho-\rho^{\prime}\|\\
&\geq \|\pi^{T}_{\rho}(\rho-\rho^{\prime})\|-\|\pi_{P}(\pi^{T}_{\rho}(\rho-\rho^{\prime}))-\pi^{T}_{\rho}(\rho-\rho^{\prime})\|-\eta\|\rho-\rho^{\prime}\|\\
&>\|\rho-\rho^{\prime}\|\left(\sqrt{1-\eta^{2}}-l-\eta\right)>\epsilon C\left(\sqrt{1-\eta^{2}}-l-\eta\right),
\end{align*}
where we used the fact that, for $0\leq l<1$, we can choose $\eta$ small enough such that $\left(\sqrt{1-\eta^{2}}-l-\eta\right)>0$. Furthermore, we can choose $C>0$ such that $C\left(\sqrt{1-\eta^{2}}-l-\eta\right)>2$. 
Since $\pi_{P}(B_{\epsilon}(\rho))=B_{\epsilon}(\pi_{P}(\rho))$, this proves the statement.\\
\\
For the converse, let $\rho,\rho^{\prime}\in\mathcal{P}$ and $\rho\neq\rho^{\prime}$. Choosing $\epsilon=\min\{\epsilon_{0},\|\rho-\rho^{\prime}\|/(2C)\}$, we find $\pi_{P}(B_{\epsilon}(\rho))\cap\pi_{P}(B_{\epsilon}(\rho^{\prime}))=\emptyset$ and thus $\pi_{P}(\rho)\neq\pi_{P}(\rho^{\prime})$. 

Finally, assume $\pi_{P}|_{\mathcal{P}}$ is not an immersion at some $\rho\in\mathcal{P}$. Let $\gamma:(-1,1)\to\mathcal{P}\subseteq H(\mathcal{H})$ be a smooth curve with $\gamma(0)=\rho$ and $\frac{d}{dt}\gamma(0)=v\in \ker\,\pi_{P}$. Let $C>0$ as in the theorem, then
\begin{align*}
2/C\leq\lim_{t\to 0}\frac{\|\pi_{P}(\gamma(t)-\rho)\|}{\|\gamma(t)-\rho\|}\leq\lim_{t\to 0}\frac{\|\pi_{P}(\gamma(t)-\rho)\|}{t/2}=2\|\pi_{P}(v)\|=0,
\end{align*}
a contradiction. Here we assumed $\|\gamma(t)-\rho\|>t/2$ which is clearly true for $t$ small enough by lemma \ref{lemtangentapp}.
\end{proof}
\begin{remark}
Note that it is essentially the constant $l$ that determines $C$. For small $l$, i.e. in the case where the tangent spaces are not steep with respect to $\sigma_P$, we can ensure that $C$ is close to $2$ (if we make $\epsilon_0$ small enough). On the other hand if $l$ is close to one $C$ has to be big and in this sense $l$ is a measure for the stability of the POVM $P$.

$\epsilon_0$ is mainly determined by the constant $\kappa$, which is more of ``global'' nature. Loosely speaking it is a measure for how bad $\mathcal{P}$ wiggles around in $H(\h)$.
\end{remark}
It is worth noting, that if $\pi_{P}$ fails to be an immersion, $C\rightarrow\infty$ for $\epsilon\rightarrow 0$ and from this point of view, stably $\mathcal{P}$-complete measurements are the ones that can be made arbitrarily precise.

This theorem transfers to the corresponding theorem for POVMs as the following corollary shows. 

\begin{corollary}\label{corstab}
Let $\mathcal{P}\subseteq \mathcal{S}(\mathcal{H})$ be a submanifold and let $P$ be a POVM with associated linear map $h_{P}$. $P$ is a smooth embedding if and only if there exists $\epsilon_{0}>0$ and $C>2$ such that for all $\epsilon$ with $0<\epsilon<\epsilon_{0}$,
\begin{align*}
\rho,\rho^{\prime}\in\mathcal{P}\text{  with }\|\rho-\rho^{\prime}\|>\epsilon C\ \Rightarrow\ h_{P}(B_{\epsilon}(\rho))\cap h_{P}(B_{\epsilon}(\rho^{\prime}))=\emptyset\ \ \ .
\end{align*}
\end{corollary}

\begin{proof}
Let $\pi_P$ be the orthogonal projection associated to $P$. Let $\lambda:=\min_{v\in \text{supp}\,\pi_{P}, \|v\|=1}\|h_{P}(v)\|$ and observe that $\lambda>0$ since $h_{P}$ is injective on the support of $\pi_{P}$. Thus $\|h_{P}(\rho-\rho^{\prime})\|=\|h_{P}\circ\pi_{P}(\rho-\rho^{\prime})\|>\lambda\|\pi_{P}(\rho-\rho^{\prime})\|$. Then, the proposition holds for $h_{P}$ by replacing $C$ with $C/\lambda$.
\end{proof}

Finally, this result also incorporates robustness against noisiness of the outcome as the following corollary shows.
\begin{corollary}
Let $\mathcal{P}\subseteq\mathcal{S}(\h)$ be a submanifold and let $P$ be a stably $\mathcal{P}$-complete POVM with associated linear map $h_{P}$. There exists $\epsilon_{0}>0$ and $C>2$ such that for all $\epsilon$ with $0<\epsilon<\epsilon_{0}$,
\begin{align*}
\mathcal{P}\cap h_{P}^{-1}(B_{2\epsilon}(h_{P}(\rho))\subseteq B_{C\epsilon}(\rho)\text{ for all }\rho\in\mathcal{P}.
\end{align*}
\end{corollary}
\begin{proof}
Let $C$, $\epsilon_{0}$ be as in corollary \ref{corstab}. Let $\rho\in\mathcal{P}$ and $\rho^{\prime}\in h_{P}^{-1}(B_{2\epsilon}(h_{P}(\rho))$ with $\|\rho-\rho^{\prime}\|>C\epsilon$. Then, $2\epsilon<\|h_{P}(\rho)-h_{P}(\rho^{\prime})\|<2\epsilon$, a contradiction.
\end{proof}

\subsection{Generalized Measurements and Smooth Embeddings}\label{sec12}
Linear measurements are clearly not sufficient to realize all smooth embeddings. More precisely, if there is a smooth embedding $\phi:\mathcal{P}\subseteq\mathcal{S}(\mathcal{H})\to\mathbb{R}^{m}$, then there need not be an $m$-dimensional POVM that is a $\mathcal{P}$-embedding. For example the set $N:=\{(x,y)\in\mathbb{R}^{2}:x^{2}+y^{2}=1,x\geq -0.5\}$ can clearly be embedded in $\mathbb{R}^{1}$, but an injective orthogonal projection has to have rank two. However, the embedding cannot get arbitrarily bad because from Whitney's embedding theorem we know that there is a $\mathcal{P}$-embedding in Euclidean space of twice the dimension of $\mathcal{P}$.

In this section we generalize our approach to measurements of the type
\begin{align*}
\text{tr}(\rho^{\otimes n}P_{i})
\end{align*}
and we show that these measurements can approximate any smooth embedding. This means that if there exists a smooth embedding $\psi:\mathcal{P}\subseteq\mathcal{S}(\h)\to\R^m$, then there is POVM of dimension $m$ that is a $\mathcal{P}$-embedding. Thus, the problem described in the beginning of this section can be circumvented by this generalized measurement scheme.

Let us fix some notation. 
\begin{definition}
A measurement $P:=\{P_{1},...,P_{m}\}$ on $k$ copies is a POVM  on $H(\mathcal{H})^{\otimes k}$. $P$ induces a linear map
\begin{align*}
h_{P}: H(\mathcal{H})^{\otimes k}&\to \mathbb{R}^{m} \\
\rho&\mapsto\big( \text{tr}(P_{1}\rho^{\otimes k}),...,\text{tr}(P_{m}\rho^{\otimes k}) \big).
\end{align*}
Let $i:H(\mathcal{H})\to H(\mathcal{H})^{\otimes k},\ \rho\mapsto\rho^{\otimes k}$. $P$ is called $\mathcal{R}$-complete for a subset $\mathcal{R}\subseteq \mathcal{S}(\mathcal{H})$ if $h_{P}\circ i|_{\mathcal{R}}=h_{P}|_{i(\mathcal{R})}$ is injective and it is called a $\mathcal{P}$-embedding for a submanifold $\mathcal{P}\subseteq \mathcal{S}(\mathcal{H})$ if $h_{P}|_{i(\mathcal{P})}$ is a smooth embedding.
\end{definition}

The following proposition makes the connection to the theory developed in the last section.
\begin{proposition}\label{propinc}
The mapping $i:H(\mathcal{H})\to H(\mathcal{H})^{\otimes k},\ \rho\mapsto\rho^{\otimes k}$ is smooth. Furthermore, for a smooth closed submanifold $\mathcal{P}\subseteq\mathcal{S}(\mathcal{H})$, $i|_{\mathcal{P}}$ is a smooth embedding.
\end{proposition}
The proof of this proposition is relegated to appendix \ref{appendixpropinc}.

\begin{remark}\label{remgenemb}
Let $\Sigma (n,k)$ be the set of $n$-dimensional operator systems on $\mathcal{B}(\mathcal{H})^{\otimes k}$. Each $n$-dimensional measurement on $k$ copies $P$ generates an operator system $\sigma_{P}=\text{span}\{P_{i}\}_{P_{i}\in P}\in\Sigma (n,k)$. If $\mathcal{P}\subseteq\mathcal{S}(\mathcal{H})$ is a closed submanifold, $i(\mathcal{P})\subseteq\mathcal{S}(\mathcal{H}^{\otimes k})=\mathcal{S}(\mathcal{H})^{\otimes k}$ is a closed submanifold by the previous proposition. So the ideas and results of the last section can be naturally applied to measurements on $k$ copies. In particular for a submanifold $\mathcal{P}\subseteq\mathcal{S}(\mathcal{H})$ the notions of $\mathcal{P}$-embedding and $\mathcal{P}$-complete naturally apply to the equivalence classes $\sigma\in\Sigma (n,k)$ of measurements on $k$ copies ($P\sim P^{\prime}\Leftrightarrow\sigma_{P}=\sigma_{P^{\prime}}$).
\end{remark}

\begin{theorem}\label{thmgenstability}
Let $\mathcal{P}\subseteq\mathcal{S}(\mathcal{H})$ be a closed submanifold and let $\sigma\in\Sigma (n,k)$ be $\mathcal{P}$-complete. Then $\sigma$ is stably $\mathcal{P}$-complete if and only if it is a $\mathcal{P}$-embedding.
\end{theorem}
\begin{proof}
By the previous remark $i(\mathcal{P})\subseteq\mathcal{S}(\mathcal{H})^{\otimes k}$ is a closed submanifold. Then the claim follows by applying \ref{thmstability} to $i(\mathcal{P})$ and $\sigma$.
\end{proof}

Choosing an othonormal basis $\{\sigma_{i}\}_{i\in \{1,...,d^{2}\}}$ of $H(\mathcal{H})$ with $\sigma_{1}=\id_\h$ gives an identification $H(\mathcal{H})\simeq\mathbb{R}^{d^2}$. Under this identification we can think of elements in $H(\h)^{\otimes k}$ as elements in $P^{k}(\mathbb{R}^{d^2})$, the vector space of polynomial functions of degree $k$ on $H(\h)\simeq\mathbb{R}^{d^2}$ \footnote{Note that by viewing $H(\mathcal{H})$ as a smooth manifold this corresponds to choosing a particular coordinate system $(x_{1},...,x_{d^{2}})$.}.

More precisely, let $\text{Sym}(H(\mathcal{H}),k)\subseteq H(\h)^{\otimes k}$ be the vector space of symmetric elements of degree $k$ in $H(\h)^{\otimes k}$. Then, use the identification $H(\h)\simeq\mathbb{R}^{d^2}$ to define a linear map 
\begin{align}\label{phi}
\phi: \text{Sym}(H(\mathcal{H}),k)\to P^{k}(\mathbb{R}^{d^2})
\end{align}
by the relation
\begin{align*}
\phi (\eta)(x)=\text{tr}\left(\eta \left(\sum_{i=1}^{n^{2}}x_{i}\sigma_{i}\right)^{\otimes n}\right)
\end{align*}
where $\eta\in \text{Sym}(H(\mathcal{H}),k)$ and $x\in \mathbb{R}^{d^2}$.
\begin{lemma}\label{lemmutl}
The mapping $\phi$ is an isomorphism.
\end{lemma}
\begin{proof}
Let $d=\dim H(\mathcal{H})$. Note that,
\begin{align*}
\dim\ \text{Sym}(H(\mathcal{H}),k)={n+k-1\choose k}=\dim\ P^{k}(\mathbb{R}^{d^2}).
\end{align*}
Then, by linearity of $\phi$, it is enough to check that $\phi$ is surjective. Under the identification $H(\mathcal{H})\simeq\mathbb{R}^{d^2}$, a basis of $P^{k}(\mathbb{R}^{d^2})$ is given by polynomials of the form $x_{i_{1}}...x_{i_{k}},\ i_{j}\in\{1,...,d\},\ i_{1}\leq...\leq i_{k}$. For each such polynomial $p=x_{i_{1}}...x_{i_{k}}$ there is a $\eta\in \text{Sym}(H(\mathcal{H}),k)$ such that $\phi(\eta)(x)=p$, namely $\eta=\sigma_{i_{1}}\cdot...\cdot\sigma_{i_{k}}$, where $\,\cdot\,$ denotes the symmetric product. 
\end{proof}

\begin{remark}
Note that every $\rho\in\mathcal{S}(\mathcal{H})$ decomposes as $\rho=\id_{\mathcal{H}}+\sum_{i=2}^{d^{2}}\sigma_{i}$ and thus $x_{1}=1$ on $\mathcal{S}(\mathcal{H})$. From this point of view $P^{k}(\mathbb{R}^{d^2})$ corresponds to $P^{\leq k}(\mathbb{R}^{d^2-1})$, the set of polynomials of degree $d\leq k$ in $x_{2},...,x_{d^{2}}$.
\end{remark}

The following lemma is the crucial ingredient of the main theorem of this section. Let $\id_{\mathcal{H}}+H(\mathcal{H})_{0}:=\{\id_\h+h:h\in H(\h)_0\}$.
\begin{lemma}\label{lemapprox}
Let $\mathcal{P}\subseteq\mathcal{S}(\mathcal{H})\subseteq \id_{\mathcal{H}}+H(\mathcal{H})_{0}\simeq \mathbb{R}^{n\times n-1}$ be a closed submanifold and $\psi:\mathcal{P}\to\mathbb{R}^{m}$ be a smooth embedding. Then, there is a $k\in\mathbb{N}$ and a map $\tilde{\psi}^{\prime}=(p_{1},...,p_{m}),\ p_{i}\in P^{\leq k}(\mathbb{R}^{n\times n-1})$, such that $\psi^{\prime}=\tilde{\psi}^{\prime}|_{\mathcal{P}}$ is a smooth embedding.
\end{lemma}
The proof of this lemma can be found in appendix \ref{appendixlemapprox}.

\begin{theorem}
Let $\mathcal{P}\subseteq\mathcal{S}(\mathcal{H})$ be a closed submanifold. There is a smooth embedding of $\mathcal{P}$ in $\mathbb{R}^{m}$ if and only if, for some $k\in\mathbb{N}$, there exists a stably $\mathcal{P}$-complete $m$-dimensional measurement on $k$ copies.
\end{theorem}
\begin{proof}
\ref{thmgenstability} gives one direction. For the other direction, let $\psi:\mathcal{P}\to\mathbb{R}^{m}$ be a smooth embedding. Then, by \ref{lemapprox}, there is a smooth embedding $\psi^{\prime}=(p_{1},...,p_{m})|_{\mathcal{P}},\ p_{i}\in P^{\leq k}(\mathbb{R}^{n\times n-1})$. $\sigma=\text{span}_\R\{ \id_{\mathcal{H}},\phi^{-1}(p_{1}),...,\phi^{-1}(p_{m})\}$ ($\phi$ form \ref{lemmutl}) is clearly an operator system whose dimension is less or equal to $m+1$. $\sigma$ is a $\mathcal{P}$-embedding because $\psi^{\prime}|_{\mathcal{P}}$ is an embedding and thus stably $\mathcal{P}$-complete by \ref{thmgenstability}.
\end{proof}

Thus, under the premise of stability, asking for the minimal dimension of a $\mathcal{P}$-complete POVM is equivalent to the related problem in differential topology of finding the smallest $m$ such that $\mathcal{P}$ can be smoothly embedded in $\mathbb{R}^{m}$.

\section{Upper and Lower Bounds for Concrete Submanifolds}\label{sec2}

In this section we obtain lower as well as upper bounds on the dimension of complete and stable POVMs on some interesting submanifolds of states. The procedure is to first relate the submanifolds to well-known homogeneous spaces and then to obtain or use existing non-immersion results for these. Upper bounds are obtained by directly constructing POVMs.\\
First, we give bounds for the set of states with fixed spectrum. Thereby we also obtain bounds for the closely related set of states with bounded rank.\\
Then, we give a brief analysis of states with an underlying unitary symmetry which is needed in the next section, where we generalize the previous results to states of fixed spectrum with an underlying symmetry.\\
Finally, we obtain bounds for the set of pure states of bipartite systems, obtained from the action of the unitary group of the second system on some fixed pure state.

In the following let $\mathcal{H}=\mathbb{C}^{n}$.

\subsection{States of Fixed Spectrum and States of Bounded Rank}\label{sec21}

First, we consider the set of states in $\mathcal{S}(\mathbb{C}^{n})$ with fixed spectrum $s=(s_{1},...,s_{n})$ \footnote{By spectrum we mean the set of eigenvalues order increasingly together with their multiplicities.} and we denote by $D_{s}:=\text{diag}(s_{1},...,s_{n})$ the diagonal matrix with entries from $s$.\\
The set of all states with spectrum $s$, $\mathcal{S}(\mathbb{C}^{n})_{s}$, is the orbit of $D_{s}$ with respect to the action $G$ of $U(n)$ on $\mathcal{S}(\mathbb{C}^{n})$ by conjugation, i.e.
\begin{align*}
\mathcal{S}^{n}_{s}:=\{U D_{s} U^{\dagger}:\ U\in U(n)\}.
\end{align*}
The isotropy group of $\rho$ under this action is $U(n_{1})\times...\times U(n_{k})$, where $n_{i}$ is the multiplicity of the $i$-th biggest eigenvalue. Note that $\sum_{j=1}^{k}n_{j}=n$. By theorem 3.62 of \cite{warner1971foundations}, factoring the orbit map over this isotropy group induces a diffeomorphism
\begin{align*}
U(n)/U(n_{1})\times...\times U(n_{k})\simeq\mathcal{S}(\mathbb{C}^{n})_{s}.
\end{align*}
Thus, $\mathcal{S}^{n}_{s}$ can be identified with a complex flag manifold.\\
In \cite{walgenbach2001lower}, Walgenbach obtains lower bounds for the immersion dimension of complex flag manifolds. To present his result, we first introduce some notation.
\begin{definition}
Let $n\in\mathbb{N}$, $k\in\{0,1,...,n\}$.\\
$\alpha(n):=$number of ones in the binary expansion of $n$,\\
$\alpha_{1}(n):=\sum_{i=0}^{n-1}\alpha(i)$,\\
$\beta(n,k):=\alpha_{1}(n)-\alpha_{1}(k)-\alpha_{1}(n-k)$,\\
\end{definition}
Let $\{n_{1},...,n_{k}\}$ be a partition of $n$. Let $K$ be some subset of $\{1,...,k\}$ and set $m=\sum_{i\in K}n_{i}$.
\begin{proposition}\cite{walgenbach2001lower}\label{propflag}
The complex flag manifold $U(n)/U(n_{1})\times...\times U(n_{k})$ cannot be immersed in Euclidean Space of dimension $4m(n-m)-2\beta(n,m)-1$ and it cannot be embedded in Euclidean space of dimension $4k(n-m)-2\beta(n,m)$.
\end{proposition}

Next, we want to obtain upper bounds on the dimension of stably $\mathcal{S}^{n}_{s}$-complete POVMs. 
Let $\sigma$ be the function that associates to each $h\in H(\mathbb{C}^{n})$ its spectrum. For $A\subseteq H(\mathbb{C}^{n})$, let $\text{Spec}(A):=\{D_{s}: s=\sigma(M), M\in A\}$  and let $G(A):=\{UMU^{\dagger}:U\in U(n),\ M\in A\}$.

\begin{lemma}\label{lemunitray}
$\Delta\mathcal{S}^{n}_{s}=G\left(\text{Spec}(\Delta\mathcal{S}^{n}_{s})\right)$ and $T\mathcal{S}^{n}_{s}=G\left(\text{Spec}(T_{D_{s}}\mathcal{S}^{n}_{s})\right)$ as sets. Furthermore, let $r$ be the biggest multiplicity of an eigenvalue in $s$, then $rank(M)<2(n-r)+1$ for $M\in \Delta\mathcal{S}^{n}_{s}\cup T\mathcal{S}^{n}_{s}$.
\end{lemma}
\begin{proof}
The first claim is essentially true by definition. For the second claim, let us compute the tangent space at $M_{s}$. Let $h\in H(\mathbb{C}^{n})$ and consider the curve $\gamma:t\mapsto e^{iht}D_{s}e^{-iht}$. The derivative at $t=0$ of this curve is then an element of $T_{D_{s}}S^{n}_{s}$ and we find
\begin{align*}
\frac{d}{dt}\gamma |_{t=0}=\frac{d}{dt}|_{t=0}e^{iht}D_{s}e^{-iht}|_{t=0}=i[h,D_{s}].
\end{align*}
In the canonical basis, these elements are of the form
\begin{align*}
\begin{pmatrix}
0 & iA & iB & \dots \\
iA^{\dagger} & 0 & iC &  \\
iB^{\dagger} & iC^{\dagger} & 0 &  \\
\vdots &  &  & \ddots
\end{pmatrix}
\end{align*}
Thus, by dimensional reasons, all elements of $T_{D_{s}}S^{n}_{s}$ are of this form. Furthermore for $U\in U(n)$, observe that
\begin{align*}
U[h,D_{s}]U^{\dagger}=[UhU^{\dagger},UD_{s}U^{\dagger}]
\end{align*}
and thus
\begin{align*}
c_U: T_{D_{s}}S^{n}_{s}&\to T_{UD_{s}U^{\dagger}}S^{n}_{s} \\
     v&\mapsto UvU^{\dagger}
\end{align*}
is an isomorphism. This proves the second claim. 
To prove the last claim, observe that for $U,V\in U(n)$
\begin{align*}
UD_{s}U^{\dagger}-VD_{s}V^{\dagger}=U(D_{s}-\lambda\id)U^{\dagger}-V(D_{s}-\lambda\id)V^{\dagger}
\end{align*}
and similarly 
\begin{align*}
[h,D_{s}]=[h,D_{s}]-[h,\lambda\id]=[h,D_{s}-\lambda\id].
\end{align*}
Choosing $\lambda$ to be the eigenvalue in $s$ with the biggest multiplicity $r$, the expressions above are differences of rank $n-r$ matrices and thus maximally of rank $2(n-r)$.
\end{proof}
\begin{remark}
It is immediate that a POVM $P$, that is injective on the set of hermitian operators with rank smaller than $r$, $\mathcal{P}_{r}$, is an $\mathcal{S}^{n}_{s}$-embedding. This is because $\Delta\mathcal{P}_{r}=\mathcal{P}_{2r}$ and thus $\Delta\mathcal{S}^{n}_{s}\cup T\mathcal{S}^{n}_{s}=\mathcal{P}_{2r}=\Delta\mathcal{P}_{r}$ by \ref{lemunitray}.
\end{remark}
As a consequence, the POVM constructed in \cite{heinosaari2013quantum} for states of bounded rank is also a $\mathcal{S}^{n}_{s}$-embedding and we obtain the following upper bounds.
\begin{proposition}\label{upflag}
Let $m$ be the biggest multiplicity of an eigenvalue in the spectrum $s$ and let $r:=n-m$. Then, there is a POVM $P$ of dimension $4r(n-r)-1$ that is a $\mathcal{S}^{n}_{s}$-embedding.
\end{proposition}
\begin{remark}
Note that for $r=n/2$ the dimension of the POVM is $4n/2(n-n/2)-1=n^2-1$. Thus, it is the trivial POVM that can identify all states and hence we also get a $\mathcal{S}^{n}_{s}$-embedding for $r>n/2$.
\end{remark}
The construction of the POVM is based on \cite{cubitt2008dimension}. The idea is to use a totally non-singular matrix, like e.g. the Vandermonde-matrix, to construct a linear subspace of $M(\mathbb{C},n)$ that just contains matrices of rank bigger than $2r$.

 \begin{table}[h]
  \setlength{\tabcolsep}{10mm}
 \begin{tabular}{lcccccccccccc}
l\textbackslash k   &         2 & 3         & 4\\
\\
                 5  & 22/34;39 \\
  \\
                 6  & 26/40;47\\ 
  \\
                 7  & 30/50;55 & 48/76;83\\
 \\
                 8  & 34/60;63 & 54/90;95\\
 \\
                 9  & 38/66;71 & 60/98;107  & 84/134;143\\
 \\
                 10 & 42/72;79 & 66/110;119 & 92/148;159\\

 \end{tabular}
 \caption{Dimension/ Lower bounds on immersion dimension \cite{walgenbach2001lower}; Upper bound on embedding dimension \ref{upflag} for $U(l+k)/U(l)\times U(1)^{k}$.}
\label{constantspectrumbounds}
\end{table}

As presented in the table \ref{constantspectrumbounds}, these results are rather close to the lower bounds of \cite{walgenbach2001lower}. Thus, the POVM of \cite{heinosaari2013quantum} gives good upper bounds on $\mathcal{S}^{n}_{s}$ and in addition we have indirectly obtained good lower bounds on the dimension of a POVM that is complete with respect to the states of bounded rank.

\subsection{States with Unitary Symmetry}\label{unisym}
Next, we shortly discuss subsets of states invariant under some unitary subgroup.

More precisely, we analyze the structure of the fix point sets of the action by conjugation $G_{H}$ of some subgroup $H\subseteq U(n)$, i.e.
\begin{align*}
\mathcal{S}(\mathbb{C}^{n})_{H}:=\{\rho\in\mathcal{S}(\mathbb{C}^{n}):\ U\rho U^{\dagger}=\rho,\forall U\in H\}.
\end{align*}

Consider the sets $\mathcal{B}(\mathbb{C}^{n})_{H}:=\{B\in\mathcal{B}(\mathbb{C}^{n}):\ UBU^{\dagger}=B,\forall U\in H\}$. $\mathcal{B}(\mathbb{C}^{n})_{H}$ is a $C^{*}$ algebra, since it is certainly a vector space and closed under the *-involution by the unitarity of $H$ and thus the structure theorem \cite{davidson1996c} yields that $\mathcal{B}(\mathbb{C}^{n})_{H}$ is unitarily equivalent to $\bigoplus_{i=1}^{k}M(n_{i},\mathbb{C})\otimes\id_{m_{i}}$.

Observe that the linear isomorphism
\begin{align}
\begin{split}
\iota:\mathcal{S}\left(\bigoplus_{i=1}^{k}M(n_{i},\mathbb{C})\right)&\to \bigoplus_{i=1}^{k}M(n_{i},\mathbb{C})\otimes \id_{m_{i}}\\
(\rho_{1},...,\rho_{k})&\mapsto (\frac{1}{m_{1}}\rho_{1}\otimes \id_{m_{i}},...,\frac{1}{m_{k}}\rho_{k}\otimes \id_{m_{k}}).
\label{eqiota}
\end{split}
\end{align}
descends to a diffeomorphism on states. Form this we immediately get the following proposition.
\begin{proposition}
There is a POVM $P$ with $\dim P=\dim \mathcal{S}(\mathbb{C}^{n})_{H}$ that is stably $\mathcal{S}(\mathbb{C}^{n})_{H}$-complete.
\end{proposition}

\subsection{Unitarily Invariant States of Fixed Spectrum}\label{sec23}
Now, given some unitary subgroup $H\subseteq U(n)$, we want to identify $G_{H}$-invariant (c.f. \ref{unisym}) states of fixed spectrum $s$\footnote{Here $s$ has to be compatible with the decomposition illustrated in section \ref{unisym}.}, i.e. 
\begin{align*}
\mathcal{S}(\mathbb{C}^{n})_{H,s}:=\{\rho\in\mathcal{S}(\mathbb{C}^{n}):\ U\rho U^{\dagger}=\rho,\forall U\in H,\ \text{spec}(\rho)=s\}.
\end{align*}
Via the map $\iota$ defined in (\ref{eqiota}), there is natural action of $U_{n_{1}}\times...\times U_{n_{k}}$ on $B(\mathbb{C}^{n})_{H}$ coming from its action on $\bigoplus_{i=1}^{k}M(n_{i},\mathbb{C})$. $\mathcal{S}(\mathbb{C}^{n})_{H}$ is then the orbit of this action on the set
\begin{align*}
\mathcal{D}:=\{(D_{M_{1}},...,D_{M_{k}}): [(M_{1})^{c_{1}}\cup...\cup (M_{k})^{c_{k}}]=s \}.
\end{align*}
Here $M_{i}$ is a multiset of order $n_{i}$ and $(M_{i})^{c_{i}}$ is the union of $c_{i}$ copies of $M_{i}$. By the same argument as in the previous section the orbit of some $\rho\in \mathcal{D}$ under $G_{H}$ is diffeomorphic to a product of complex flag manifolds $\prod_{i=1}^{k}U(n_{i})/\prod_{j=1}^{k_{j}}U(n^{i}_{j})$. Since $\mathcal{D}$ is clearly finite, $\mathcal{S}(\mathbb{C}^{n})_{H,s}$ is a disjoint union of products of complex flag manifolds. Thus, it is enough to look at one of these components at a time to get non-immersion results.

For some component, let $m_{i}$ be the number associated to the $i$-th factor in the product, that is constructed just like the number $m$ for \ref{propflag}.
\begin{proposition}\label{propflagprod}
The product of complex flag manifolds $\prod_{i=1}^{k}U(n_{i})/\prod_{j=1}^{k_{j}}U(n^{i}_{j})$ cannot be immersed in Euclidean Space of dimension $\sum_{i=1}^{k}(4m_{i}(n-m_{i})-2\beta(n,m_{i}))-1$ and it cannot be embedded in Euclidean space of dimension $\sum_{i=1}^{k}4m_{i}(n-m_{i})-2\beta(n,m_{i})$.
\end{proposition}
The proof of this result can be found in appendix \ref{appendixb}. Of cource, \ref{upflag} also transfers to this situation and gives upper bounds on the dimension of stably $\mathcal{S}(\mathbb{C}^{n})_{H,s}$-complete POVMs.

\subsection{Bob-Unitary Orbit}\label{sec24}
Let $\alpha\in\mathcal{H}_{A}\otimes\mathcal{H}_{B}$, $\langle\alpha|\alpha\rangle= 1$. In this section we investigate pure states of the form
\begin{align*}
\mathcal{S}_{B}(\alpha):=\{|\beta\rangle\langle\beta|\in\mathcal{S}(\mathcal{H}_{A}\otimes\mathcal{H}_{B}):\beta=(\id\otimes U)\alpha, U\in U(\mathcal{H}_{B})\}.
\end{align*}
Let $\{e_{1},...,e_{\dim \mathcal{H}_{A}}\}$, $\{f_{1},...,f_{\dim \mathcal{H}_{B}}\}$ be orthonormal bases of $\mathcal{H}_{A}$ respectively $\mathcal{H}_{B}$ such that
\begin{align*}
\alpha=\sum_{i=1}^{r}\alpha_{i} e_{i}\otimes f_{i}
\end{align*}
is a Schmidt decomposition, where $r$ is the Schmidt rank of $\alpha$. Then, $\mathcal{S}_{B}(\alpha)$ is diffeomorphic to the projective Stiefel manifold $PW_{n,r}$. In order to see this, note that $(\id\otimes U(\mathcal{H}_{B}))\alpha$ is diffeomorphic to the complex Stiefel manifolds $W_{n,r}:=\{m\in\mathbb{C}^{d\times r}:m^{*}\cdot m=\id\}$ via
\begin{align*}
i:&(\id\otimes U(\mathcal{B}))\alpha\to W_{n,k},\\
&M_{i,j}(\sum_{i=1}^{r}\alpha_{i} e_{i}\otimes Uf_{i}):=\langle Uf_{i}|f_{j}\rangle
\end{align*}
Factoring both sides over the free action of the cyclic group $S^{1}\subseteq\mathbb{C}$, $m\mapsto z\cdot m$, then yields the desired map \cite{warner1971foundations}.

In order to state the main result of this section we introduce two functions,
\begin{definition} Let $n,k\in\mathbb{N}$ and $n\ge k$.
\begin{align*}
&\text{1.}\ N(n,k):=\min\{ n-k<i\leq n:\ {n \choose i}\pmod{2}\equiv 1\},\\
&\text{2.}\ \sigma(n,k):=2\cdot\max\{0\leq i<N(n,k):\ {nk+i-1 \choose i}\pmod{2}\equiv 1\}.
\end{align*}
\end{definition}
\begin{proposition}\label{propbob}
Let $\alpha\in\mathcal{H}_{A}\otimes\mathcal{H}_{B},\ \langle\alpha|\alpha\rangle= 1$, with Schmidt rank $k$ and $n=\dim \mathcal{H}_{B}$. Then $\mathcal{S}_{B}(\alpha)$ cannot be immersed in Euclidean space of dimension $(2n-k)k-1+\sigma(n,k)$ and cannot be embedded in Euclidean space of dimension 
$(2n-k)k-1+\sigma(n,k)+1$.
\end{proposition}
The proof of this result is very similar to \cite{barufatti1994obstructions} and can be found in appendix \ref{appendixa}.  This non-immersion result is obtained deploying a standard approach based on the dual Stiefel-Whitney class of the tangent bundle \cite{milnor1974characteristic}.\\

For $k=1$ the complex projective Stiefel manifold is just the complex Projective space, so in this chase we can compare the result obtained here to the upper bounds of Milgram \cite{milgram1967immersing}, which are known to be close to optimal. Table \ref{tabwn1} shows these bounds for some dimensions. For $n=2^{k}$, $k\in\mathbb{N}$, the dual Stiefel-Whitney classes give no obstructions, whereas they essentially equal Milgram's result in \cite{milgram1967immersing} for $n=2^{k}+1$, $k\in\mathbb{N}$, and hence are close to optimal in this case. 
 
  \begin{table}[h]
  \setlength{\tabcolsep}{3mm}
 \begin{tabular}{lccccccccccccccc}
 2 & 6 & 6 & 14 & 14 & 14 & 14 & 30 & 30 & 30 & 30 & 30 & 30 & 30 & 30 & 62 \\
 2 & 6 & 8 & 14 & 16 & 21 & 22 & 30 & 32 & 37 & 38 & 45 & 46 & 52 & 52 & 62 
 \end{tabular}
 \caption{Lower bounds on immersions of $W_{n,1}\simeq P\mathbb{C}^{n}$ for $n=2,...,17$. In the first row the result is obtained from the dual Stiefel-Whitney classes in the second row the results of \cite{milgram1967immersing} are presented.}
 \label{tab:wn1}
 \end{table}
 
 \begin{table}[h]
  \setlength{\tabcolsep}{3mm}
 \begin{tabular}{lcccccccccccc}
  n\textbackslash r   & 2&3&4&5&6&7&8&9&10&11&12&13\\
  2 & 2  & 3 \\
  3 & 6  & 7 & 8\\ 
  4 & 6  & 11 & 14 & 15\\
  5 & 14 & 19 & 22 & 23 & 24\\
  6 & 14 & 27 & 30 & 31 & 34  & 35\\
  7 & 14 & 27 & 38 & 39 & 46  & 47  & 48\\
  8 & 14 & 27 & 38 & 47 & 54  & 59  & 62  & 63\\
  9 & 30 & 43 & 54 & 63 & 70  & 75  & 78  & 79  & 80\\
 10 & 30 & 51 & 54 & 63 & 86  & 91  & 94  & 95  & 98  & 99\\
 11 & 30 & 55 & 72 & 79 & 86  & 107 & 110 & 111 & 118 & 119 & 120\\
 12 & 30 & 59 & 78 & 79 & 102 & 107 & 126 & 127 & 134 & 139 & 142 & 143\\
 \end{tabular}
 \caption{Lower bounds on immersion dimension of $PW_{n,r}$ obtained from dual Stiefel-Whitney classes \ref{propbob}.}
\label{tab:wnk}
\end{table}

In \cite{astey2000parallelizability}, it is shown that $PW_{n,n}$ and $PW_{n,n-1}$ is parallelizable for $n\neq2$ and thus can be immersed in Euclidean space of codimension one by a result of Hirsch \cite{hirsch1959immersions}. For $PW_{4,k}$ and $PW_{8,k}$, there are no obstructions because the dual Stiefel-Whitney classes vanish for $nk=q2^{r}$, $q,r\in\mathbb{Z}$ and $N(n,k)<2^{r}$ \cite{barufatti1994obstructions}.\\ 
The dual Stiefel-Whitney classes do not generally give good obstructions, but can be supplemented by other methods. In Table \ref{tab:wnk} these bounds are presented for some explicit scenarios. Another approach to the non-immersion problem is due to \cite{atiyah1988immersions}. In a similar vein, another method is given to obtain non-immersion results, with the exterior powers $\gamma_{i}$ of $KO(X)$ playing the role of the Stiefel-Whitney classes. Both of these methods are worked out and compared in \cite{barufatti1994obstructions}.\\

Next, we give upper bounds on the dimension of an $\mathcal{S}_{B}(\alpha)$-embedding, presenting two different approaches.

The first approach is based on the upper bounds obtained for states of fixed spectrum. The problem is split into determining the minor obtained by tracing over $\mathcal{H}_{A}$ and afterwards determining the relative phases.

Before stating the upper bounds, let us first prove the following lemma which will be useful later on.
\begin{lemma}\label{lemtr}
Let $\alpha:=\sum_{i=1}^{r} \lambda_{i} e_{i}\otimes f_{i}$, $O\in H(\mathcal{H}_{A})$, $S\in H(\mathcal{H}_{B})$, $U\in U(\mathcal{H}_{B})$ and $P_{\alpha}:=\sum_{i=1}^{r} \lambda_{i}|e_{i}\rangle\langle f_{i}|$. Then $\text{tr}(O\otimes USU^{\dagger} |\alpha\rangle\langle \alpha|)=\text{tr}((P_{\alpha}U)^{\dagger}O^{T}(P_{\alpha}U)S)$.
\end{lemma}
\begin{proof}
The prove of this is a straightforward computation.
\begin{align*}
\text{tr}(O\otimes USU^{\dagger} |\alpha\rangle\langle \alpha|)&=\langle\alpha|O \otimes USU^{\dagger}|\alpha\rangle
=\sum_{i,j=1}^{r}\lambda_{i}\lambda_{j}\langle e_{i}|\otimes\langle f_{i}|O\otimes USU^{\dagger}|e_{j}\rangle\otimes |f_{j}\rangle\\
&=\sum_{i,j=1}^{r}\lambda_{i}\lambda_{j}\langle e_{i}|O |e_{j}\rangle\langle f_{i}|USU^{\dagger}| f_{j}\rangle
=\sum_{i,j=1}^{r}\lambda_{j}\langle e_{j}|O^{T}|e_{i}\rangle \lambda_{i}\langle f_{i}|USU^{\dagger}|f_{j}\rangle\\
&=\sum_{i,j=1}^{r}\langle f_{i}|f_{j}\rangle\lambda_{j}\langle e_{j}|O^{T}P_{\alpha}USU^{\dagger}| f_{i}\rangle
=\sum_{i=1}^{r}\langle f_{i}|U^{\dagger}P_{\alpha}^{\dagger}O^{T}P_{\alpha}US| f_{i}\rangle\\
&=\text{tr}((P_{\alpha}U)^{\dagger}O^{T}(P_{\alpha}U)S)
\end{align*}

\end{proof}
The following proposition is motivated by a method to embed Lie groups in Euclidean space, introduced in \cite{rees1971some}.
\begin{proposition}\label{up1}
Let $\alpha\in\mathcal{H}_{A}\otimes\mathcal{H}_{B}$ with Schmidt rank $k$ and $n=\dim \mathcal{H}_{B}$. Then, there is a $\mathcal{S}_{B}(\alpha)$-embedding of dimension $4r(n-r)-1+4n-5$ for $r<n/2$ and $n^{2}-1+4n-5$ for $r\geq d/2$.
\end{proposition}
\begin{proof}
First, note that by lemma \ref{lemtr} we can assume w.l.o.g that $\lambda_{i}\neq\lambda_{j}$ for $1\leq i,j \leq r,\ i\neq j$, because this can always be achieved by choosing $O$ appropriately.

The idea is to take advantage of the natural projection $\pi:PW_{n,r}\to \frac{U(n)}{U(1)^{r}\times U(n-r)}$, which just amounts to choosing $O=\id$ in lemma \ref{lemtr}. More precisely, for $O=\id$ we get
\begin{align*}
P_{\alpha}OP_{\alpha}^{\dagger}=\begin{pmatrix}
\lambda_1 &  &  & &  \\
 & \ddots &  & &  \\
 &  & \lambda_r & &  \\
 &  &  & 0 & \\
 &  &  &   & \ddots
\end{pmatrix}
\end{align*}
and thus we are in the situation discussed in the last section with the isotropy group given by $U(1)^{r}\times U(n-r)$. This means, the projected state can be embedded in dimension $4r(d-r)-1$ using the POVM of \ref{upflag}. Let us call this map $\phi_{1}$.

Let $v\in \mathbb{C}^{r}$ be the vector with a one in every entry and consider the map
\begin{align*}
\psi:U(n)/U(n-r)&\to \mathbb{C}^{n}\\
U&\mapsto U(v\oplus 0).
\end{align*}
This is clearly well defined and also observe that $\psi$ descents to a map $\tilde{\psi}:PW_{n,r}\simeq U(n)/(U(n-r)\times U(1))\to P\mathbb{C}^{n}$. Let $U,V\in U(n)$ with $U\sim V$ in $U(n)/U(1)^{r}\times U(n-r)$, i.e. $U=V(D\oplus\id)(\id\oplus W)$ with $D\in U(1)^{r}$ and $W\in U(n-r)$. Then, $Uv=\lambda Vv$ just has a solution for $U\sim V$ in $U(n)/(U(1)\times U(n-r))$. To see this, note that
\begin{align*}
U(v\oplus 0)&=\lambda V(v\oplus 0)\\
V(D\oplus\id)(\id\oplus W)(v\oplus 0)&=\lambda V(v\oplus 0)\\
(D\oplus\id)(v\oplus 0)&=\lambda (v\oplus 0)\\
Dv&=\lambda v
\end{align*}
and thus $D=\lambda \id$ is the only solution.

Hence, supplementing the embedding above by $\tilde{\psi}$ guarantees injectivity and the only problem left, is embedding $P\mathbb{C}^{n}$.

In terms of lemma \ref{lemtr} $\tilde{\psi}$ corresponds to choosing $P_{\alpha}O^{T}P_{\alpha}^{\dagger}=v^\dagger v\oplus 0$, the projective version of $v\oplus 0$. Then choose $S$ according to the POVM of \cite{heinosaari2013quantum} to obtain an embedding in Euclidean space. Let us call this mapping $\phi_{2}$

The map $\phi:=(\phi_{1},\phi_{2})$ is clearly smooth as well as injective and thus a topological embedding by compactness of $\mathcal{S}_{B}(\alpha)$. 
From lemma \ref{lemtr} it is easy to see that $\frac{d}{dt}|_{t=0}(\id\otimes Ue^{iHt})|\alpha\rangle\langle\alpha| (\id\otimes e^{-iHt}U^{\dagger})$ for $H=h\oplus 0$ and $h\neq \id$ diagonal gives a $n-1$ dimensional subspace $V_{U}$ of the tangent space at $\alpha_{U}:=(\id\otimes U)|\alpha\rangle\langle\alpha| (\id\otimes U^{\dagger})$. $V_{U}$ is clearly in the kernel of $d\phi_{1}|_{\alpha_{U}}$. Thus, by dimensional reasoning, it is enough to see that $d\phi_{2}|_{\alpha_{U}}$ is injective on $V_{U}$. Since the POVM of \cite{heinosaari2013quantum} can identify all rank one matrices, it is enough to see that $h\oplus 0\mapsto [v\dagger v\oplus 0,h\oplus 0]$ is injective for $h\neq \id$ diagonal. This can be easily verified by a direct computation.
\end{proof}
\begin{remark}
This result is best for $r$ close to $n$, so in particular for $PW_{n,n}$, the set of maximally entangled states. For $PW_{n,n}$ we obtain an embedding in Euclidean space of codimension $4n-5$.

Furthermore, note that in the context of quantum process tomography, the result for maximally entangled states can be used to identify a unitary time evolution. Preparing a certain maximally entangled state, the POVM given above can identify unitary processes up to a phase, i.e. with $\mathcal{O}(n^{2})$ measurements.
\end{remark}

The second approach relies on the direct construction of an $\mathcal{S}_{B}(\alpha)$-embedding.
\begin{proposition}\label{up2}
Let $\alpha\in\mathcal{H}_{A}\otimes\mathcal{H}_{B}$ with Schmidt rank $r$ and $n=\dim \mathcal{H}_{B}$. Then, there is a $\mathcal{S}_{B}(\alpha)$-embedding of dimension $2nr+2n-3$.
\end{proposition}
\begin{proof}
First, note that w.l.o.g we can assume $\lambda_{1}=...=\lambda_{r}=1$, as can be easily seen from lemma \ref{lemtr}.

Then, for $O=|e_{i}\rangle\langle e_{j}|, S=|f_{k}\rangle\langle f_{l}|$ we obtain
\begin{align*}
\text{tr}((P_{\alpha}U)^{\dagger}O^{T}(P_{\alpha}U)S)=\langle f_{i}|U|f_{k}\rangle \langle f_{j}|U|f_{l}\rangle^{*},
\end{align*}
so from this point of view any linear combination of such products of elements of $P_{\alpha}U$ determines an operator, that need not be hermitian, and a set of such equations determines an operator system (here we think of $P_{\alpha}U$ as a matrix in the $\{e_{i}\}_{i\in\{1,...r\}},\ \{f_{l}\}_{l\in\{1,...n\}}$ basis). It is worth noting that an equation not corresponding to a non-hermitian operator actually corresponds to two operators in the operator system, namely its hermitian and anti-hermitian part.

Let $M_{n(i-1)+j}(U):=\langle e_{i}|P_{\alpha}U|f_{j}\rangle$ for $i\in\{1,...,r\},\ j\in\{1,...,n\}$ and $M_{k}:=0$ for $k>nr$. For $k\in\{1,...,nr+n-1\}$, define operators $\tilde{G}_{k}$ via the equations
\begin{align*}
G_{k}(U):=\sum_{i=1,i\leq k+1-i}^{n}M_{i}(U)M_{k+1-i}^{*}(U).
\end{align*}
Then, the operator system $\sigma_{G}$ spanned by the $\tilde{G}_{k}$ is an $\mathcal{S}_{B}(\alpha)$-embedding. It is clear that the dimension of $\sigma_{G}$ is $2nr+2n-3$, noting that non-hermitian operators count twice. 

Let $U,V\in U(\mathcal{H}_{B})$. In order to prove injectivity, we have to show that if $G_{k}(U)=G_{k}(V)$, then there is a $\phi\in\mathbb{R}$ such that $P_{\alpha}U=e^{i\phi}P_{\alpha}V$. First, observe that for $M_{1}(U)=...=M_{k}(U)=0$ we have $k\leq n$ because $P_{\alpha}U$ has full rank. Let $m$ be the smallest number such that $M_{m}$ does not vanish. Then the claim is clearly true for all $j<m$. Now, let $l>m$ and assume that the claim holds for all $j\leq l$, then 
\begin{align*}
G_{m+l}(U)&=\sum_{i=1,i\leq m+l+1-i}^{n}M_{i}(U)M_{m+l+1-i}^{*}(U)\\
&=M_{m}(U)M_{l+1}(U)^{*}+\sum_{i=m+1,i\leq m+l+1-i}^{n}M_{i}(U)M_{m+l+1-i}^{*}(U)\\
&=M_{m}(V)e^{i\phi}M_{l+1}(U)^{*}+\sum_{i=m+1,i\leq m+l+1-i}^{n}M_{i}(V)M_{m+l+1-i}^{*}(V)\\
&=G_{m+l}(V),
\end{align*}
thus $e^{-i\phi}M_{l+1}(U)=M_{l+1}(V)$.

To conclude the proof, we need to show that the measurement constructed above is an immersion. For $h\in H(\mathcal{H}_{B})$, $U\in U(\mathcal{H}_{B})$ and and define a curve $\gamma(t):=(\id\otimes e^{iUhU^{\dagger}t}U)|\alpha\rangle\langle\alpha|(\id\otimes U^{\dagger}e^{-iUhU^{\dagger}t})$. The derivative of this curve yields tangent vectors at $(\id\otimes U)|\alpha\rangle\langle\alpha|(\id\otimes U^{\dagger})$ and by lemma \ref{lemtr} an effect operator $O\otimes S$ maps these to
\begin{align*}
&\frac{d}{dt}|_{t=0}\text{tr}((P_{\alpha}e^{iUhU^{\dagger}t}U)^{\dagger}O^{T}(P_{\alpha}e^{-iUhU^{\dagger}t}U)S)\\
&=i\text{tr}(P_{\alpha}^{\dagger}O^{T}P_{\alpha}U[h,S]U^{\dagger}).
\end{align*}
For $k\in\{1,...,nr+n-1\}$, this yields the equations
\begin{align*}
F_{k}(U,h)=\sum_{i=1,i\leq m+l+1-i}^{n}M_{i}(U)M_{k+1-i}^{*}(Uh)-M_{i}(Uh)M_{k+1-i}^{*}(U).
\end{align*}

Observe that $F_{k}(U,h)=F_{k}(U,h+c\id)$ holds for each $c\in\mathbb{C}$. Furthermore, it is easy to see that for every $k\in\{1,...,nr\}$ and every $h\in H(\mathcal{H}_{B})$ there is a $\lambda\in\mathbb{C}$ such that $M_{k}(U(\lambda\id+h))=0$ and thus we can assume w.l.o.g that $M_{k}(Uh)=0$.

Let $m$ be the smallest number such that $M_{m}(U)$ does not vanish and assume w.l.o.g $M_{m}(Uh)=0$. Let $l\in\{1,...,nr+m-1\}$. It is easy to see that the vanishing of these equations for all $i\leq l$ implies that $M_{j}(Uh)=0$ for $j\leq l+1-k$ and thus, we obtain injectivity on a real vector space of dimension $2nr-r^{2}-1$.
\end{proof}

 \begin{table}[h]
 \begin{tabular}{lcccccccccccc}
n\textbackslash r   &                   5 & 9                   & 17                  & 65\\
\\
  5  & 24/24;40;57 \\
  \\
  9  &  64/70;112;105      & 80/80;112;177\\ 
  \\
 17  & 144/166;302;201     & 224/238;352;337     & 288/288;352;609\\
 \\
 65  & 624/742;1454;777    & 1088/1198;2270;1297 & 1920/2014;3518;2337 & 4224/4224;4480;8577\\
 \\
 129 & 1264/1510;2990;1545 & 2240/2478;4830;2577 & 4096/4318;8126;4641 & 12544/12670;17152;17025\\

 \end{tabular}
 \caption{Dimension/Lower bounds on immersion dimension \ref{propbob}; First upper bound on embedding dimension \ref{up1}; Second upper bound on embedding dimension \ref{up2} for $PW_{n,r}$.}
\label{stiefelbounds}
\end{table}

In table \ref{stiefelbounds} both of these methods are compared. It is clear that the embedding in \ref{up2} works best for $k/m\ll 1$ , because this approach does not take the orthogonality of the $f_{i}$ into account. The embedding in \ref{up1} works best for $k/m\sim 1$, because just in this case the projected state can be determined efficiently.

\appendix
\section{Technical appendix}\label{sec3}
\subsection{Proof of Theorem \ref{thmstability}}\label{appendixthmstability}
Before we give the proof, let us first fix some notion.

Let $\mathcal{SO}(H(\mathcal{H}))$ be the orthogonal group on the inner product space $H(\h)$. The generalized Pauli basis together with the identity $\id_{\mathcal{H}}$ gives an identification of $H(\h)\simeq \R^{d^2}$ and the Hilbert-Schmidt inner product induces the standard inner product on $\R^{d^{2}}$. Thus $\mathcal{SO}(H(\mathcal{H}))$ can be identified with $\mathcal{SO}(\R^{d^{2}})$, the standard orthogonal group on $\R^{d^{2}}$. Denote by $\mathcal{SO}(H(\mathcal{H}))_{v}:=\{O\in \mathcal{SO}(H(\mathcal{H})):O v=v\}$ the stabilizer subgroup of $v\in H(\mathcal{H})\simeq \R^{d^2}$. Note that for $O\in\mathcal{SO}(H(\mathcal{H}))_{\id_{\mathcal{H}}}$ and $\sigma\in\Sigma (n)$, we have $O\sigma\in\Sigma (n)$. Thus there is an action of $\mathcal{SO}(H(\mathcal{H}))_{\id_{\mathcal{H}}}$ on $\Sigma(n)$,
\begin{align*}
\Sigma(n)\times\mathcal{SO}(H(\mathcal{H}))_{\id_{\mathcal{H}}}&\to \Sigma(n)\times\Sigma(n)\\
(\sigma,O) &\mapsto (\sigma,O\sigma).
\end{align*}
The geometric intuition of thinking of operator systems as planes in $H(\h)\simeq \R^{d^2}$ is essential for the following proof. Then, for $O\in\mathcal{SO}(H(\mathcal{H}))_{\id_{\mathcal{H}}}$, $O\sigma$ is just a rotated plane and the intuition is that for small rotations these operator systems are close.

\begin{proof}
Let $\sigma\in\Sigma (n)$ be stably $\mathcal{P}$-complete and $\pi_{\sigma}$ be the associated orthogonal projection. Furthermore let $SH(\mathcal{H}):=\{B\in H(\mathcal{H}):\|B\|=1\}$ be the unit sphere in $H(\mathcal{H})$.  Assume by contradiction that $\sigma$ is not a $\mathcal{P}$-embedding, i.e. $\pi_{\sigma}$ is not an immersion. Then, since $d(\pi_{\sigma})_{\rho}=\pi_{\sigma}|_{T_\rho\mathcal{P}}$, there exists a point $\rho\in\mathcal{P}$ and a smooth curve $\gamma:(-1,1)\to\mathcal{P}$ with $\gamma (0)=\rho$ and $v=\dot{\gamma}(0)\in\sigma^{\bot},\ v\in SH(\mathcal{H})$. The idea is that $\gamma(t)\approx \rho+vt$ for small $t$ and to obtain a contradiction we construct for a point $\rho^{\prime}=\gamma(t^{\prime})\approx \rho+vt^{\prime}$ an operator system $\sigma^{\prime}$ with $\pi_{\sigma^{\prime}}(\rho^{\prime}-\rho)=0$. This procedure is presented in figure \ref{fig2}.
\begin{figure}[ht]
\includegraphics[width=7cm]{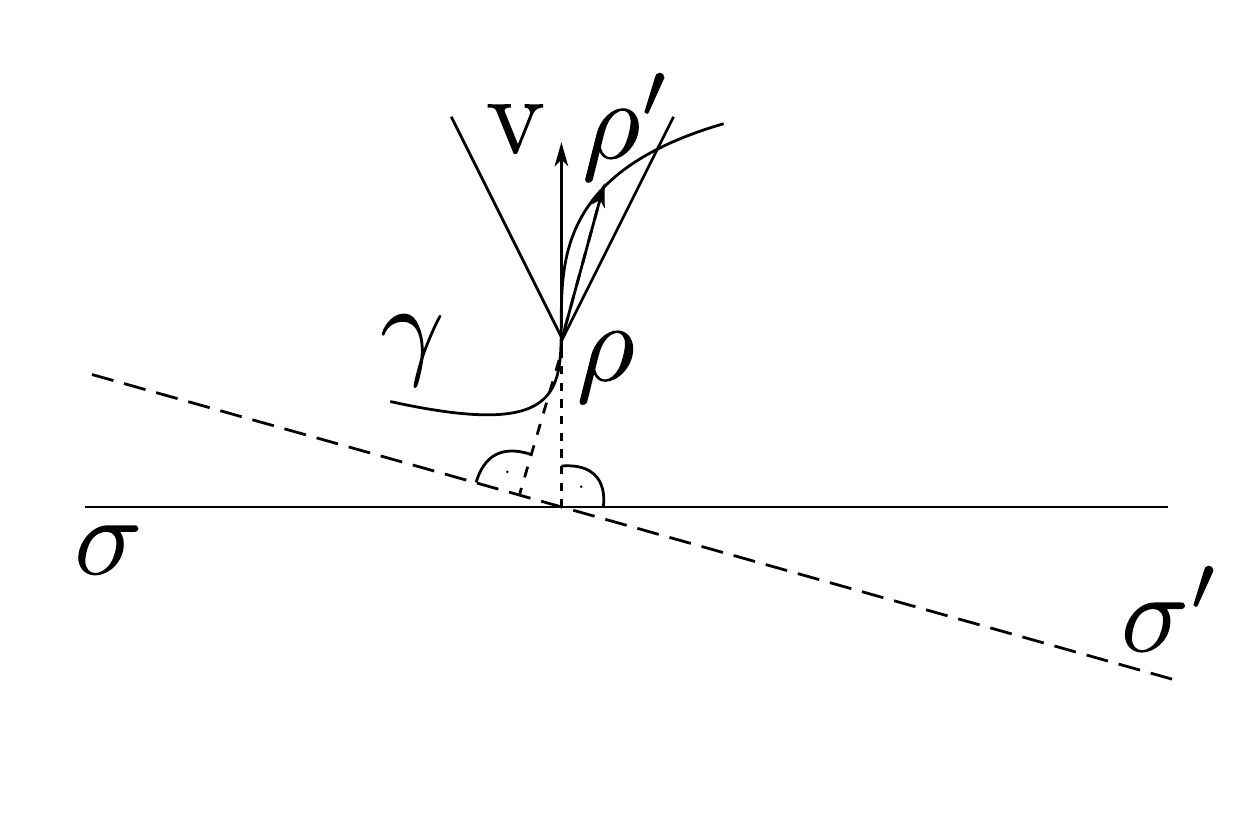}

\caption{\label{fig2} This figure shows the curve $\gamma$ with $\gamma(0)=\rho$, $\gamma(t^{\prime})=\rho^{\prime}$ and $\frac{d}{dt}\gamma(0)=v\in \sigma^{\bot}$ together with the operator system $\sigma^{\prime}$ that is constructed such that $\rho^{\prime}-\rho\in(\sigma^{\prime})^{\bot}$.}
\end{figure}

More precisely, we prove that for each $\delta>0$ there is a $t\in(0,1)$ such that $\gamma(t)\in V_{\rho}^{\delta}:=\{\rho+\lambda\cdot Ov:\ \lambda>0,\ O\in\mathcal{SO}(H(\mathcal{H})),\ \|\id-O\|_{op}<\delta\}$.

First, we prove that $V_{\rho}^{\delta}$ is open. Note that the left action
\begin{align*}
\mathcal{SO}(H(\mathcal{H}))\times SH(\mathcal{H})&\to SH(\mathcal{H})\times SH(\mathcal{H})\\
(O,v)&\mapsto(Ov,v)
\end{align*}
is smooth and transitive. Thus, the orbit map $\beta_{v}: \mathcal{SO}(H(\mathcal{H}))\to SH(\mathcal{H}),\ O\mapsto Ov$ is smooth and factors over the natural projection $\pi:\mathcal{SO}(H(\mathcal{H}))\to \mathcal{SO}(H(\mathcal{H}))/\mathcal{SO}(H(\mathcal{H}))_{v}$ (theorem 3.62 of \cite{warner1971foundations}), i.e. $\beta_v=\tilde{\beta}_{v}\circ\pi$ with $\tilde{\beta}_{v}$ a diffeomorphism. In particular $\beta_v$ is open because $\pi$ is open \footnote{For an open set $A\subseteq \mathcal{SO}(H(\mathcal{H}))$ we find $\pi^{-1}(\pi(A))=\mathcal{SO}(H(\mathcal{H}))_{v}\cdot A=\bigcup_{O\in\mathcal{SO}(H(\mathcal{H}))_{v}}O(A)$. So $\pi^{-1}(\pi(A))$ is open for any open set $A\subseteq \mathcal{SO}(H(\mathcal{H}))$ and thus $\pi$ is open.}.

Since $\beta_v$ is open there is an $\eta>0$ such that $\emptyset\neq B_{\rho+v}(\eta)\cap B_{\rho}(1)\subseteq V_{\rho}^{\delta}$. By possibly shrinking $\eta$ we can even assume that $\overline{B}_{\rho+v}(\eta)\subseteq V_{\rho}^{\delta}$ because of the conic structure of $V_{\rho}^{\delta}$. It follows that $B_{\rho+sv}(s\eta)\subseteq V_{\rho}^{\delta}$ for $s>0$.

Then,
\begin{align*}
\frac{\|\gamma(t)-(\rho+tv)\|}{t}=\frac{\|\gamma(t)-\gamma(0)-tv\|}{t}\to 0\ \text{as}\ t\to 0,
\end{align*}
whereas we find for the distance $d(\rho+tv,\partial B_{\rho+tv}(t\eta))/t:=\inf_{\rho'\in\partial B_{\rho+tv}(t\eta)}\|\rho-\rho'\|/t=\eta$. So by continuity of the norm there is a $t>0$ such that $\gamma((0,t))\subseteq V_{\rho}^{\delta}$.
\\

But then, for every $\delta>0$, there is an $O\in\mathcal{SO}(H(\mathcal{H}))$, a $\lambda>0$ and a $t>0$ with
\begin{align*} 
d(\sigma,O\sigma)&=\|\pi_{\sigma}-\pi_{O\sigma}\|_{op}\\
&=\|\pi_{\sigma}-O\pi_{\sigma}O^{-1}\|_{op}\\
&\leq \|\pi_{\sigma}-\pi_{\sigma}O^{-1}\|_{op}+\|\pi_{\sigma}-O\pi_{\sigma}\|_{op}\\
&\leq 2\|\id-O\|_{op}<2\delta
\end{align*}
such that $\gamma(t)-\gamma(0)=\lambda Ov\neq 0$. But then, $\pi_{O\sigma}(\gamma(t)-\gamma(0))=O\pi_{\sigma}(O^{-1}(\gamma(t)-\gamma(0)))=\lambda O\pi_{\sigma}(v)=0$ by assumption on $v$. Also note that $\langle\id_{\mathcal{H}},\gamma (t)-\gamma (0)\rangle:=tr[\id_{\mathcal{H}}\left(\gamma (t)-\gamma (0)\right)]=0$ and $\langle\id,T_{\rho}\mathcal{P}\rangle=0$ and thus we can choose $O\in\mathcal{SO}(H(\mathcal{H}))_{\id_{\mathcal{H}}}$. So $O\sigma$ is an operator system but it is not $\mathcal{P}$-complete, contradicting the stability of $\sigma$.\\

Conversely, suppose $\sigma$ is a $\mathcal{P}$-embedding. \ref{linap} states, that there is an $\epsilon>0$ such that every $\sigma^{\prime}\in\Sigma (n)$ with $\sup_{B\in H(\mathcal{H}),\ \|B\|\leq 1}\|\pi_{\sigma}(B)-\pi_{\sigma^{\prime}}(B)\|<\epsilon$
is a $\mathcal{P}$-embedding and thus in particular $\mathcal{P}$-complete..
\end{proof}

\subsection{Proof of Lemma \ref{lemtangentapp}}\label{appendixlemtangentapp}
The following proof uses geometric concepts and is based on the identification of the tangent spaces with planes in $H(\h)$. 
\begin{proof}
$\pi^{T}$ is smooth, as can be easily seen in local coordinates. The mapping
\begin{align*}
\psi: \mathcal{P}\times\mathcal{P}&\to \R \\
	(\rho^{\prime},\rho)&\mapsto \|\pi^{T}_{\rho_{\prime}}-\pi^{T}_{\rho}\|_{op}
\end{align*}
is clearly continuous as a composition of continuous mappings and thus, for every $\eta>0$, there is an open neighbourhood $N_{\rho_{0}}$ of $\rho_{0}\in\mathcal{P}$ such that $\psi(\rho,\rho_{0})<\eta/4$ for all $\rho\in N_{\rho_{0}}$. Let $\nu_0>0$ and let $B_{5\nu_0}(\rho_{0})$ be the open ball of radius $5\nu_0$ around $\rho_{0}$, such that $B_{5\nu_0}(\rho_{0})\cap\mathcal{P}$ is contained in $N_{\rho_{0}}$.

Let $\rho\in B_{\nu_0}(\rho_{0})\cap\mathcal{P}$. Then, for all $\tilde{\rho}\in B_{4\nu_0}(\rho)\cap\mathcal{P}$, we find
\begin{align}
\psi(\rho,\tilde{\rho})<\psi(\rho,\rho_{0})+\psi(\tilde{\rho},\rho_{0})<\eta/2.
\end{align}
Let $\rho^{\prime}\in \partial B_{\epsilon}(\rho)\cap\mathcal{P}$, $0<\epsilon<\nu_0$. Furthermore, let $\gamma:[0,\lambda]\to\mathcal{P}\subseteq \mathcal{S}(\mathcal{H})$ be a geodesic that connects $\rho$ and $\rho^{\prime}$ with $\gamma(0)=\rho$ and $\frac{d}{dt}\gamma(t)|_{t=0}=v, \|v\|=1$. Since $(\rho,v)\mapsto\frac{d^2}{dt^2}\gamma(t)|_{t=0}=\frac{d^2}{dt^2}\exp(\rho,vt)|_{t=0}$ is a smooth function from the compact set $\mathcal{P}\times S^{\dim\mathcal{P}-1}$ to $H(\mathcal{H})$, there is $k\geq 0$ such that $k:=\max_{(\rho,v)\in\mathcal{P}\times S^{\dim\mathcal{P}-1}}\|\frac{d^2}{dt^2}exp(\rho,vt)|_{t=0}\|$. It follows from the geodesic equation
\begin{align*}
\pi^{T}_{\gamma(t)}\left(\frac{d^2}{dt^2}\gamma(t)\right)&=0 \\
\pi^{T}_{\gamma(0)}\left(\frac{d^2}{dt^2}\gamma(t)\right)+(\pi^{T}_{\gamma(t)}-\pi^{T}_{\gamma(0)})\left(\frac{d^2}{dt^2}\gamma(t)\right)&=0 \\
\pi^{T}_{\gamma(0)}\left(\frac{d}{dt}\gamma(t)\right)+\int_{0}^{t}dt^{\prime}(\pi^{T}_{\gamma(t^{\prime})}-\pi^{T}_{\gamma(0)})\left(\frac{d^2}{{dt^{\prime}}^2}\gamma(t^{\prime})\right)&=v\\
\pi^{T}_{\gamma(0)}\left(\gamma(t)\right)+\int_{0}^{t}dt^{\prime}\int_{0}^{t^{\prime}}dt^{\prime\prime}(\pi^{T}_{\gamma(t^{\prime\prime})}-\pi^{T}_{\gamma(0)})\left(\frac{d^2}{{dt^{\prime\prime}}^2}\gamma(t^{\prime\prime})\right)&=vt+\pi^{T}_{\gamma(0)}\left(\rho\right).
\end{align*}
However, for $t\in[0,4\nu_0]$,
\begin{align*}
&\left\|\int_{0}^{t}dt^{\prime}(\pi^{T}_{\gamma(t^{\prime})}-\pi^{T}_{\gamma(0)})\left(\frac{d^2}{{dt^{\prime}}^2}\gamma(t^{\prime})\right)\right\|\\
\leq&\int_{0}^{t}dt^{\prime}\left\|(\pi^{T}_{\gamma(t^{\prime})}-\pi^{T}_{\gamma(0)})\left(\frac{d^2}{{dt^{\prime}}^2}\gamma(t^{\prime})\right)\right\|\\
<&\int_{0}^{t}dt^{\prime}k\eta/2=kt\eta/2\leq 2 \eta k \nu_0.
\end{align*}
as well as
\begin{align*}
&\left\|\int_{0}^{t}dt^{\prime}\int_{0}^{t^{\prime}}dt^{\prime\prime}(\pi^{T}_{\gamma(t^{\prime\prime})}-\pi^{T}_{\gamma(0)})\left(\frac{d^2}{{dt^{\prime\prime}}^2}\gamma(t^{\prime\prime})\right)\right\|\\
\leq&\int_{0}^{t}dt^{\prime}\int_{0}^{t^{\prime}}dt^{\prime\prime}\left\|(\pi^{T}_{\gamma(t^{\prime\prime})}-\pi^{T}_{\gamma(0)})\left(\frac{d^2}{{dt^{\prime\prime}}^2}\gamma(t^{\prime\prime})\right)\right\|\\
<&\int_{0}^{t}dt^{\prime}\int_{0}^{t^{\prime}}dt^{\prime\prime}k\eta/2= k t^{2}\eta/4\leq 4 \eta k (\nu_0)^{2}.
\end{align*}
Thus, we find
\begin{align*}
\left\|\pi^{T}_{\gamma(0)}\left(\gamma(t)-(vt+\rho)\right)\right\|&<k t^{2}\eta/4 < 4 \eta k (\nu_0)^{2}
\end{align*}
and
\begin{align*}
\left\|\frac{d}{dt}\gamma(t)-v\right\|&=\left\|\pi^{T}_{\gamma(t)}\left(\frac{d}{dt}\gamma(t)\right)-\pi^{T}_{\gamma(0)}\left(\frac{d}{dt}\gamma(t)\right)+\pi^{T}_{\gamma(0)}\left(\frac{d}{dt}\gamma(t)\right)-v\right\|\\
&\leq \left\|\pi^{T}_{\gamma(t)}\left(\frac{d}{dt}\gamma(t)\right)-\pi^{T}_{\gamma(0)}\left(\frac{d}{dt}\gamma(t)\right)\right\|+\left\|\pi^{T}_{\gamma(0)}\left(\frac{d}{dt}\gamma(t)\right)-v\right\|\\
&\leq \eta/2+2 \eta k \nu_0,\\
\end{align*}
Note that $\gamma$ stays inside $ B_{5\nu_0}(\rho_{0})$. Furthermore, for $\eta$ and $\nu_0$ small enough, $\gamma$ intersects $\partial B_{2\nu_0}(\rho_{0})$ and $\gamma$ intersects $\partial B_{\epsilon}(\rho)$ close to radial. In particular it follows that  $t\mapsto\|\gamma(t)-\rho\|$ is strictly increasing as long as $\gamma$ stays inside $B_{2\nu_0}(\rho_{0})$. Then, each geodesic intersects $\partial B_{\epsilon}(\rho)$ exactly once before it passes through $\partial B_{2\nu_0}(\rho_{0})$ for each $\epsilon$ with $0<\epsilon<\nu_0$. Let $K_{\rho}$ be the connected component of $\mathcal{P}\cap B_{\nu_0}(\rho)$ containing $\rho$ and let $K_{\rho_{0}}$ be the connected component of $\mathcal{P}\cap B_{2\nu_0}(\rho_{0})$ containing $\rho_{0}$. The above reasoning implies that $K_{\rho}\cap B_{\epsilon}(\rho)$ is connected and that $K_{\rho_{0}}\cap B_{\tilde{\epsilon}}(\rho_{0})$ is connected for $0<\tilde{\epsilon}<2\nu_0$.

We can assume w.l.o.g. that  $B_{\nu_0}(\rho)\cap\mathcal{P}$ is connected for all $\rho\in B_{\nu_0}(\rho_{0})$. Because if not, we can shrink $\nu_0$  until $B_{2\nu_0}(\rho_{0})$ contains a single connected component. To see this, shrink $\nu_0$ such that $0<\nu_0<\max_{\rho\in \overline{K_{\rho_{0}}},\rho^{\prime}\in(\overline{B_{2\nu_0}}(\rho_{0})\cap\mathcal{P})-\overline{K_{\rho_{0}}}}\|\rho-\rho^{\prime}\|$.

We then find for  $0<\epsilon<\nu_0$ and $t\in[0,4\epsilon]$,
\begin{align*}
&\|\pi^{T}_{\gamma(0)}\left(\gamma(t)-\rho\right)\|=\|\pi^{T}_{\gamma(0)}\left(\gamma(t)-(vt+\rho)+vt\right)\|\\
>&\|vt\|-\|\pi^{T}_{\gamma(0)}\left(\gamma(t)-(vt+\rho)\right\|>t-4 \eta k \epsilon^{2}.
\end{align*}

Hence, $\gamma(t)=\rho^{\prime} \text{ and }\gamma([0,t])\subseteq B_{2\nu_0}(\rho_{0})\Rightarrow t<\epsilon+4 \eta k \epsilon^{2}$.

Choosing $\nu_0$ such that $4 k  \nu_0<1$, we find for the component of $\rho^{\prime}-\rho=\gamma(t)-\gamma(0)$ normal to $T_{\rho}\mathcal{P}$,
\begin{align*}
\left\|\pi^{N}_{\rho}(\rho^{\prime}-\rho)\right\|=&\left\|\int_{0}^{t}dt^{\prime}\frac{d}{dt^{\prime}}\gamma(t^{\prime})-\pi_{\gamma(0)}^{T}\left(\frac{d}{dt^{\prime}}\gamma(t^{\prime})\right)\right\|\\
\leq&\int_{0}^{t}dt^{\prime}\left\|\pi_{\gamma(t^{\prime})}^{T}\left(\frac{d}{dt^{\prime}}\gamma(t^{\prime})\right)-\pi_{\gamma(0)}^{T}\left(\frac{d}{dt^{\prime}}\gamma(t^{\prime})\right)\right\|\\
<&\int_{0}^{t}dt^{\prime}\eta/2=t\eta/2<\eta(\epsilon+4 k \epsilon^2)/2\leq\epsilon\eta(1+4 k \nu_0)/2<\eta\epsilon=\eta\|\rho-\rho^{\prime}\|.
\end{align*}

Now, for a given $\eta>0$, construct such neighbourhoods for all $\rho\in\mathcal{P}$ to obtain a cover of $\mathcal{P}$ by open sets, $\{B_{\nu_{\rho}}(\rho)\cap\mathcal{P}\}_{\rho\in\mathcal{P}}$. By compactness of $\mathcal{P}$ there is a finite subcover $\{B_{\nu_{\rho_{i}}}(\rho_{i})\cap\mathcal{P}\}_{i\in I}$ and set $\epsilon:=\min_{i\in I}\nu_{\rho_{i}}$.
\end{proof}
\begin{remark}
Note that the proof of this lemma shows that for $0<\tilde{\epsilon}<\epsilon$, $B_{\tilde{\epsilon}}(\rho)\cap\mathcal{P}$ is connected.
\end{remark}
\subsection{Proof of Proposition \ref{propinc}}\label{appendixpropinc}

\begin{proof}
To see that $i$ is smooth, choose an orthonormal basis $\{\sigma_{i}\}_{i\in I}$ of hermitian operators for $H(\mathcal{H})$. Expansion in this basis gives global coordinates on $H(\mathcal{H})$. Expansion in $\{\sigma_{i_{1}}\otimes...\otimes\sigma_{i_{k}}\}_{i_{1},...,i_{k}\in I}$ gives global coordinates on $H(\mathcal{H})^{\otimes k}$. In these coordinates $i$ is just a polynomial and hence smooth.\\
$\mathcal{P}$ is a smooth submanifold and since $\mathcal{P}\subseteq\mathcal{S}(\mathcal{H}) \subseteq H(\mathcal{H})$, $i|_{\mathcal{P}}$ is smooth. We prove that $i|_{\mathcal{P}}$ is injective. Note that $\rho^{\otimes k}=\sigma^{\otimes k}\ \text{iff}\ \sigma=a\cdot\rho,\ a^{k}=1$. But then $\sigma,\rho\in H(\mathcal{H})$ implies $a\in\{-1,1\}$ and the positivity of both $\sigma$ and $\rho$ yields $a=1$.\\
Finally, $i|_{\mathcal{P}}$ is an immersion. To see this let $\rho\in\mathcal{P}$ and $v\in T_{\rho}\mathcal{P}$. Furthermore, let $\gamma:(-1,1)\to H(\mathcal{H})$ be a smooth curve with $\gamma (0)=\rho$ and $\frac{d}{dt}\gamma(0)=v$. First, observe that for $k=2$,
\begin{align*}
di_{\rho}v&=\frac{d}{dt}|_{t=0}(i\circ\gamma)=\lim_{t\to 0}\frac{\gamma (t)^{\otimes 2}-\gamma (0)^{\otimes 2}}{t}\\
&=\lim_{t\to 0}\frac{(\gamma (t)-\gamma (0))\otimes \gamma (t)-(\gamma (0)-\gamma (t))\otimes \gamma (0)}{t}=\rho\otimes v+v\otimes\rho.
\end{align*}
This inductively generalizes to arbitrary $k\in\mathbb{N}$ and we get
\begin{align*}
di_{\rho}v=\sum_{i=1}^{k}\rho^{\otimes i-1}\otimes v\otimes\rho^{\otimes k-i}.
\end{align*}
This is zero if and only if $v=0$, what can be easily seen by orthogonally decomposing $v$ with respect to $\rho$.

Finally, since $i|_{\mathcal{P}}$ is smooth, injective and an immersion, it is a smooth embedding by the compactness of $\mathcal{P}\subseteq\mathcal{S}(\mathcal{H})$.
\end{proof}
\begin{remark}
Note that the proof shows that $i|_{H(\h)-\{0\}}$ is an immersion.
\end{remark}

\subsection{Proof of Lemma \ref{lemapprox}}\label{appendixlemapprox}
In this lemma we prove that, for a submanifold $\mathcal{P}\subseteq \mathcal{S}(\h)$, every smooth embedding $\psi:\mathcal{P}\to \R^m$ can be approximated by a polynomial map $F:\id_{\mathcal{H}}+H_0(\h)\simeq \R^{d^2-1}\to \R^m$.
Let us first state lemma 1.3 of \cite{hirsch1976differential}.
\begin{lemma}\cite{hirsch1976differential}\label{lememb}
Let $U\subseteq\mathbb{R}^{n}$ be open and $W\subseteq U$ be open with compact closure $\overline{W}\subseteq U$. Let $f:U\to \mathbb{R}^{n}$ be a smooth embedding. There exists  $\epsilon>0$ such that if $g:U\to\mathbb{R}^{n}$ is smooth and 
\begin{align*}
\|D_{\alpha}g(x)-D_{\alpha}f(x)\|_{2}<\epsilon\ \ and\ \ \|g(x)-f(x)\|_{2}<\epsilon
\end{align*}
for all $x\in W$, $|\alpha|=1$, then $g|_{W}$ is an embedding.
\end{lemma}
Now we give the proof of lemma \ref{lemapprox}.
\begin{proof}
Note that $\psi^{\prime}=\tilde{\psi}^{\prime}|_{\mathcal{P}}$ is smooth because it is a restriction of smooth functions to a smooth submanifold. $\psi$ can be extended to a compactly supported smooth map $\tilde{\psi}$ on $\id_{\mathcal{H}}+H(\mathcal{H})_{0}\simeq \R^{n^2-1}$ and let $K\subseteq\R^{n^2-1}$ be a compact set containing $\text{supp}\,\tilde{\psi}$. In the following we make use of an approximation result given by theorem 1 in \cite{bagby2002multivariate}. The relevant part for us is that for every $\eta>0$, there is a $k\in\mathbb{N}$ such that $\tilde{\psi}$ and $d\tilde{\psi}$ can be approximated simultaneously by a map $\tilde{\psi}^{\prime}=(p_{1},...,p_{i}),\ p_{i}\in P^{\leq k}(\mathbb{R}^{n\times n-1})$, i.e $\sup_{x\in K}\|\tilde{\psi}(x)-\tilde{\psi}^{\prime}(x)\|_{2}<\eta$ and $\sup_{(x,v)\in TK, \|v\|\leq 1}\|d\tilde{\psi}_{x}(v)-d\tilde{\psi}^{\prime}_{x}(v)\|_{2}<\eta$.

Let $\{(\phi_{i},W_{i})\}_{i\in I}$ be a finite atlas on $\mathcal{P}$ and let $\tilde{B}(r_{i}):=\phi_{i}(B(r_{i}))\subseteq W_{i}$ be the image of an open ball of radius $r_{i}$ around the origin such that $\bigcup_{i\in I} \tilde{B}(r_{i}/2)=\mathcal{P}$ and $\overline{\tilde{B}(r_{i})}\subseteq W_{i}$. Applying lemma 1.3 of \cite{hirsch1976differential} to $\psi\circ\phi_{i}$, $\phi^{-1}(W_{i})$ and $B(r_{i})$, we obtain for each $i\in I$ an $\epsilon_{i}>0$ such that for all $\psi^{\prime}$ with $\sup_{B(r_{i})}\|D_{\alpha}(\psi\circ\phi_{i})-D_{\alpha}(\psi^{\prime}\circ\phi_{i})\|_{2}<\epsilon_{i}$, $|\alpha|\leq 1$, $\psi^{\prime}|_{\tilde{B}(r_{i})}$ is an embedding. 

For $|\alpha|=1$, we have
\begin{align*}
&\sup_{x\in B(r_{i})}\|D_{\alpha}(\psi\circ\phi_{i})(x)-D_{\alpha}(\psi^{\prime}\circ\phi_{i})(x)\|_{2}\\
=&\sup_{x\in B(r_{i})}\|(d\psi_{\phi(x)}-d\psi^{\prime}_{\phi(x)})\circ D_{\alpha}\phi_{i}(x)\|_{2}\\
\leq& \sup_{(x,v)\in T\tilde{B}(r_{i}),\|v\|\leq 1}\|(d\psi_{x}-d\psi^{\prime}_{x})v\|_{2} \sup_{B(r_{i})}\|D_{\alpha}\phi_{i}\|_{2}\\
\leq& \sup_{(x,v)\in TK,\|v\|\leq 1}\|(d\tilde{\psi}_{x}-d\tilde{\psi^{\prime}}_{x})v\|_{2} \sup_{B(r_{i})}\|D_{\alpha}\phi_{i}\|_{2}.
\end{align*}

For $|\alpha|=1$,  let $\kappa_{i,\alpha}:=\epsilon_{i}/\sup_{B(r_{i})}\|D_{\alpha}\phi_{i}\|_{2}$ and let $\epsilon:=\min_{i,|\alpha|=1}\{\epsilon_{i},\kappa_{i,\alpha}\}$. Then, for every  $\tilde{\psi}^{\prime}$ with $\sup_{x\in K}\|\tilde{\psi}(x)-\tilde{\psi}^{\prime}(x)\|_{2}<\epsilon$ and $\sup_{(x,v)\in TK, \|v\|\leq 1}\|d\tilde{\psi}_{x}(v)-d\tilde{\psi}^{\prime}_{x}(v)\|_{2}<\epsilon$, $\tilde{\psi}^{\prime}|_{\tilde{B}(r_{i})}$ is an embedding for all $i\in I$ and by theorem 1 of \cite{bagby2002multivariate} such a $\tilde{\psi}^{\prime}$ exists for some $k\in\mathbb{N}$. 

Finally we show that there is an $\epsilon\geq 0$ such that for every $\tilde{\psi}^{\prime}$ with $\|D_{\alpha}\tilde{\psi}-D_{\alpha}\tilde{\psi}^{\prime}\|<\epsilon$, $\psi^{\prime}|_{\mathcal{P}}$ is injective. Then $\psi^{\prime}$ is both an immersion and injective and thus a smooth embedding by the compactness of $\mathcal{P}$.\\
$\tilde{B}(r_{i})^{c}:=\mathcal{P}-\tilde{B}(r_{i})$ and $\overline{\tilde{B}(r_{i}/2)}$ are closed and therefore compact as closed subsets of a compact set. Then, by the continuity of the norm, $\eta_{i}:=\min_{q\in \tilde{B}(r_{i})^{c},\ p\in \overline{\tilde{B}(r_{i}/2)}}\|\psi(p)-\psi(q)\|$ exists and it is bigger than $0$ because $\psi$ is injective and $\tilde{B}(r_{i})^{c}\cap \tilde{B}(r_{i}/2)=\emptyset$. By possibly shrinking $\epsilon$, make sure that $\epsilon\leq \frac{1}{4}\min_{i\in I}\eta_{i}$. 

Assume $\psi^{\prime}(p)=\psi^{\prime}(q)$, $p,q\in\mathcal{P}$. Since $\psi^{\prime}|_{\mathcal{P}}$ is an embedding around $p$, there is an $i\in I$ such that $p\in \tilde{B}(r_{i}/2)$ and $p\in \tilde{B}(r_{i})^{c}$. Thus,
\begin{align*}
&\|\psi^{\prime}(p)-\psi^{\prime}(q)\|_{2}=\|\psi(p)-\psi(q)+\psi^{\prime}(p)-\psi(p)+\psi(q)-\psi^{\prime}(q)\|_{2}  \\
\geq & |\|\psi(p)-\psi(q)\|_{2}-\|\tilde{\psi}^{\prime}(p)-\tilde{\psi}(p)+\tilde{\psi}(q)-\tilde{\psi}^{\prime}(q)\|_{2} \\
\geq & 4\epsilon-2\epsilon> 0,
\end{align*}
a contradiction.
\end{proof}
A direct consequence of this proof is the following corollary, which was used in the third section.
\begin{corollary}\label{linap}
Let $\mathcal{P}\subseteq\mathcal{S}(\h)$ be a submanifold and let $L:H(\h)\to\mathbb{R}^m$ be a linear map such that $L|_{\mathcal{P}}$ is a smooth embedding. Then, there is an $\epsilon>0$ such that for every linear map $L^{\prime}:H(\h)\to\mathbb{R}^m$ with $\sup_{v\in H(\h),\|v\|\leq 1}\|(L-L^{\prime})v\|_{2}<\epsilon$, $L^{\prime}|_{\mathcal{P}}$ is a smooth embedding.
\end{corollary}
\begin{proof}
Let $K\subseteq H(\h)$ be a compact set containing $\mathcal{P}$. Furthermore let $b:H(\h)\to H(\h)$ be a smooth and compactly supported bump function which equals the identity on $K$. Then, the proof of \ref{lemapprox} shows that there is an $\eta>0$ such that for every smooth map $\psi:H(\h)\to\R^m$ with $\sup_{x\in K}\|\psi(x)-(L\circ b)(x)\|_{2}<\epsilon$ and $\sup_{(x,v)\in TK,\|v\|\leq 1}\|d\psi_{x}(v)-d(L\circ b)_{x}(v)\|_{2}<\eta$, $\psi|_{\mathcal{P}}$ is a smooth embedding. But for $\psi$ linear we find
\begin{align*}
&\sup_{(x,v)\in TK,\|v\|\leq 1}\|d\psi_{x}(v)-d(L\circ b)_{x}(v)\|_{2}\\
=&\sup_{(x,v)\in TK,\|v\|\leq 1}\|\psi(v)-L(v)\|_{2}\\
=&\sup_{v\in H(\h),\|v\|\leq 1}\|\psi(v)-L(v)\|_{2}
\end{align*}
and 
\begin{align*}
&\sup_{x\in K}\|\psi(x)-(L\circ b)(x)\|_{2}\\
\leq&\sup_{x\in K}\|x\|\sup_{v\in H(\h),\|v\|\leq 1}\|\psi(v)-L(v)\|_{2}.
\end{align*}
Thus, the claim holds for $\epsilon:=\eta/\sup_{x\in K}\|x\|>0$.
\end{proof}

\subsection{Proof of Proposition \ref{propflagprod}}\label{appendixb}
In order to prove \ref{propflagprod}, let us first fix some notation. Let $X$ be an oriented smooth compact manifold and $K(X)$ be the K-ring of $X$, i.e. the ring of equivalence classes of complex vector bundles on $X$, where $E\sim E^{\prime}$ if $E+n\simeq E^{\prime}+m$ (An introductory text on this topic is e.g. \cite{hatcher2003vector}). \footnote{Here $m$ denotes the $m$-dimensional trivial bundle.} Let $p_{i}(X):=p_{i}(TX)\in H^{2i}(X,\mathbb{Q})$ be the image of $i$-th rational Pontryagin class evaluated on the tangent bundle $TX$. Furthermore, let $\hat{A}(p_{1},...,p_{n})$ be the $\hat{A}$-genus, i.e. the genus associated to the power series $\frac{\sqrt{z}/2}{sinh(\sqrt{z}/2)}$\cite{hirzebruch1995topological}.\footnote{This means that by construction the identitiy $p(E\oplus F)=p(E)p(F)$ of the total rational Pontryagin class transfers to the $\hat{A}$-genus.}  Let $\text{ch}:K(X)\to H^{*}(X,\mathbb{Q})$ be the Chern class and let $\text{ch}(X):=\text{ch}(K(X))\subseteq H^{*}(X,\mathbb{Q})$. For $z:=\sum_{i=0}^{\infty}z^{2i}\in H^{*}(X,\mathbb{Q})$, with $z_{2j}\in H^{2i}(X,\mathbb{Q})$, let $z^{(t)}:=\sum_{i=0}^{\infty}z^{2i}t^{j}$ for $t\in\mathbb{Q}$. Note that $(yz)^{(t)}=y^{(t)}z^{(t)}$.\\
For $z\in\text{ch}(X)$, $d\in H^{2}(X,\mathbb{Q})$, $t\in\mathbb{Q}$ we define the Hilbert polynomial in $t$ to be $H_{X,z,d}(t):=(z^{t}e^{d/2}\hat{A}(p(X)))[X]$, where $[X]$ is the fundamental class of $X$\cite{borel1959characteristic,atiyah1959quelques}. Furthermore for $q\in\mathbb{Q}$ let $\nu_2(q):=\text{exponent of }2\text{ as primefactor of }q$. The result of Walgenbach in \cite{walgenbach2001lower} is based on the result of Mayer in \cite{mayer},
\begin{theorem}\cite{mayer}
Let $X$ be a $2n$-dimensional compact oriented smooth manifold and $H$ be the Hilbert polynomial associated with $d\in H^{2}(X;Z)$ and $z\in\text{ch}(X)$.\\
Then $X$ cannot be immersed in Euclidean space of dimension $-2\nu_{2}(H(\frac{1}{2}))-1$ and cannot be embedded in Euclidean space of dimension $-2\nu_{2}(H(\frac{1}{2}))$.
\end{theorem}
Walgenbach obtains his results by computing $H(\frac{1}{2})$ for some $d\in H^{2}(X;Z)$ and $z\in\text{ch}(X)$ using combinatorical methods.
With the following lemma, \ref{propflagprod} follows directly by observing that $\nu_{2}(a\cdot b)=\nu_{2}(a)+\nu_{2}(b)$. Let $X_{1}$, $X_{2}$ be $2n$-dimensional compact oriented smooth manifolds and let $\pi_{i}:X_{1}\times X_{2}\to X_{i}$, $i=1,2$ , be the canonical projections.
\begin{lemma}
For $z_{i}\in\text{ch}(X_{i})$, $d_{i}\in H^{2}(X_{i},\mathbb{Q})$ and $[X_{i}]=(\pi_{i})_{*}[X_{1}\times X_{2}]$, $i=1,2$, let $z:=\pi_{1}^{*}(z_{1})\pi_{2}^{*}(z_{2})$, $d:=\pi_{1}^{*}(d_{1})+\pi_{2}^{*}(d_{2})$. Then,
\begin{align*}
H_{X_{1}\times X_{2},z,d}(t)=H_{X_{1},z_{1},d_{1}}(t)H_{X_{2},z_{2},d_{2}}(t).
\end{align*}
\end{lemma}
\begin{proof}
First, note that $z\in K(X_{1}\times X_{2})$, since 
\begin{align*}
\pi_{1}^{*}(z_{1})\pi_{2}^{*}(z_{2})=\pi_{1}^{*}(\text{ch}(E_{1}))\pi_{2}^{*}(\text{ch}(E_{2}))=\text{ch}(\pi_{1}^{*}(E_{1}))\text{ch}(\pi_{2}^{*}(E_{2}))=\text{ch}(\pi_{1}^{*}(E_{1})\oplus\pi_{2}^{*}(E_{2})).
\end{align*}
Furthermore note that
\begin{align*}
p(X_{1}\times X_{2})&=p(T(X_{1}\times X_{2}))=p(\pi_{1}^{*}TX_{1}\oplus \pi_{2}^{*}TX_{2}))\\
&=p(\pi_{1}^{*}(TX_{1}))p(\pi_{2}^{*}(TX_{2}))=\pi_{1}^{*}(p(TX_{1}))\pi_{2}^{*}(p(TX_{2})).
\end{align*}
This, together with the fact that all cohomology classes involved are even dimensional and hence commute, yields
\begin{align*}
H_{X_{1}\times X_{2},z,d}(t)&=\left(z^{(t)}e^{d/2}\hat{A}(X_{1}\times X_{2})\right)[X_{1}\times X_{2}]\\
&=\left(\pi_{1}^{*}(z_{1})^{(t)}\pi_{2}^{*}(z_{2})^{(t)}e^{\pi_{1}^{*}(d_{1}/2)+\pi_{2}^{*}(d_{2}/2)}\pi_{1}^{*}(\hat{A}(TX_{1}))\pi_{2}^{*}(\hat{A}(TX_{2})\right)[X_{1}\times X_{2}]\\
&=\left(\pi_{1}^{*}\left(z_{1}^{(t)}e^{\pi_{1}^{*}(d_{1})/2)}\hat{A}(X_{1})\right)\pi_{2}^{*}\left(z_{2}^{(t)}e^{\pi_{2}^{*}(d_{2})/2)}\hat{A}(X_{2})\right)\right)[X_{1}\times X_{2}]\\
&=\left(z_{1}^{(t)}e^{\pi_{1}^{*}(d_{1})/2)}\hat{A}(X_{1})\right)[X_{1}]\left(z_{2}^{(t)}e^{\pi_{2}^{*}(d_{2})/2)}\hat{A}(X_{2})\right)[X_{2}]\\
&=H_{X_{1},z_{1},d_{1}}(t)H_{X_{2},z_{2},d_{2}}(t).
\end{align*}
\end{proof}

\subsection{Proof of Proposition \ref{propbob}}\label{appendixa}
Let $X$ be a smooth compact $n$-manifold with tangent bundle $TX$. Let $\phi :X\to\mathbb{R}^{n+k}$ be an immersion and let $NX$ be the normal bundle, i.e. $TX\oplus NX\simeq n+k$. Furthermore let $\omega$ be the total Stiefel-Whitney class. Let us state the following  well-known result.
\begin{proposition}\cite{milnor1974characteristic}
Let $i$ be the degree of $\overline{\omega}(X)=\omega(NX)\in H^{*}(X,\mathbb{Z}_{2})$. Then $X$ cannot be immersed in Euclidean space of dimension $n+i$ and cannot be embedded in Euclidean space of dimension $n+i+1$.
\end{proposition}
In order to use this result, we need to compute $\overline{\omega}(PW_{n,k})$. The following is similar to \cite{barufatti1994obstructions}, where the dual Stiefel-Whitney class of the real projective Stiefel manifolds is computed. Let $L$ be the complex line bundle associated to the $U(1)$-principal bundle $W_{n,k}\to PW_{n,k}$ and let $x$ be the $\text{mod}\ 2$ Euler class of $L$\footnote{I.e. the image of the Euler class of $L$ under the coefficient homomorphism $H^{*}(PW_{n,k},\mathbb{Z})\to H^{*}(PW_{n,k},\mathbb{Z}_{2})$}. In \cite{astey1999cohomology}, the cohomology ring $H^{*}(PW_{n,k},\mathbb{Z}_{2})$ for $k<n$ is found to be
\begin{align*}
\mathbb{Z}_{2}[x]/(x^{N})\oplus\Lambda(y_{n-k+1},...,y_{n}),
\end{align*}
with $y_{i}\in H^{2i-1}(PW_{n,k},\mathbb{Z}_{2})$.
It is shown in \cite{astey2000parallelizability}, that $TPW_{n,k}$ is stably isomorphic to $nkL^{*}$, where $L^{*}$ is regarded as a real vector bundle. Hence $\omega (TPW_{n,k})=\omega (L^{*})^{nk}$ and we obtain
\begin{align*}
\omega (NPW_{n,k})&=\overline{\omega}(TPW_{n,k})=\omega (L^{*})^{-nk}.
\end{align*}
Since the odd Stiefel-Whitney classes of complex vector bundles (regarded as real vector bundles) vanish \cite{hatcher2003vector}, and the Euler class is mapped to the top Stiefel Whitney class under the coefficient homomorphism $H^{*}(PW_{n,k},\mathbb{Z})\to H^{*}(PW_{n,k},\mathbb{Z}_{2})$\cite{hatcher2003vector}, we get $\omega (L)=\omega (L^{*})=1+x$. Thus,
\begin{align*}
\omega(NPW_{n,k})&=(1+x)^{-nk}\\
&=\sum_{i=1}^{\infty}(-1)^{j}{nk+j-1 \choose j}x^{j}.
\end{align*}
We now want to find
\begin{align*}
\gamma(m,k)= \text{the biggest}\ j\ \text{such that the coefficient of}\ x^{j}\  \text{does not vanish in}\ \mathbb{Z}_{2}[x]/(x^{N})
\end{align*}
Since we factor over the ideal generated by $x^{N}$, we clearly have $j\leq N(n,k)$. Passing to $\text{mod} 2$, we get $2\gamma(m,k)=\sigma(n,k)$. This proves \ref{propbob}.

\bibliographystyle{unsrt}
\bibliography{bibliography}

\end{document}